\documentclass[sigconf, nonacm]{acmart}
\usepackage[linesnumbered,ruled,vlined]{algorithm2e}
\newcommand\vldbpagestyle{plain} 
\renewcommand\footnotetextcopyrightpermission[1]{}
\usepackage{tabularx}
\usepackage{array}
\usepackage{booktabs}
\usepackage{amsmath} 
\usepackage{enumitem}
\usepackage{subcaption}
\usepackage{multirow}   
\usepackage{xcolor}
\usepackage[normalem]{ulem}
\usepackage{soul}
\usepackage[most]{tcolorbox}
\usepackage{balance}
\usepackage{mathtools}
\usepackage{pifont}
\usepackage{setspace}

\newcommand{\yin}[1]{\textcolor{black}{#1}}
\definecolor{revRoneColor}{RGB}{176,36,48}
\definecolor{REVRONECOLOR}{RGB}{176,36,48}

\newcommand{\revRoneStart}{}
\newcommand{\reviseend}{}
\long\def\revRone#1{#1}

\long\def\comment#1{}

\newtheorem{example}{Example}
\newif\ifrnsgfullversion
\rnsgfullversiontrue
\newcommand{\rnsgappendixref}[1]{\ifrnsgfullversion Appendix~\ref{#1}\else the full version\fi}
\newcommand\vldbauthors{\authors}
\newcommand\vldbtitle{\shorttitle} 

\newcommand\vldbdoi{XX.XX/XXX.XX}
\newcommand\vldbpages{XXX-XXX}
\newcommand\vldbvolume{14}
\newcommand\vldbissue{1}
\newcommand\vldbyear{2026}

\usepackage{setspace}

\begin{document}
\pagestyle{empty}
\raggedbottom

\title{RNSG: A Range-Aware Graph Index for Efficient Range-Filtered Approximate Nearest Neighbor Search}

\author{Zhiqiu Zou$^1$, Ziqi Yin$^2$, Rong-Hua Li$^1$, Hongchao Qin$^1$, Qiangqiang Dai$^1$, Guoren Wang$^1$}
\affiliation{
\begin{center}
    {$^1$Beijing Institute of Technology, Beijing, China} \quad
    {$^2$Nanyang Technological University, Singapore}~~~~~
\end{center}
}
\email{fall@bit.edu.cn, ziqi003@e.ntu.edu.sg, {lironghuabit, wanggrbit}@126.com, {qhc.neu, qiangd66}@gmail.com}



\begin{abstract}
Range-filtered approximate nearest neighbor (RFANN) search is a fundamental operation in {modern data systems.} 
Given a set of objects, each with a vector 
and a numerical attribute, an RFANN query retrieves the nearest neighbors to a query vector among those objects whose numerical attributes fall within the range \yin{specified by the query}. \yin{Existing} state-of-the-art methods for RFANN search often require constructing multiple range-specific graph indexes to achieve high query performance, which incurs significant indexing overhead. To address this, we first establish a novel graph indexing theory, the range-aware relative neighborhood graph (RRNG), which jointly considers spatial and attribute proximity. We prove that the RRNG satisfies two crucial properties: (1) \textit{monotonic searchability}, which
ensures correct nearest neighbor retrieval via beam search; and (2) \textit{structural heredity}, which guarantees that any range-induced sub-graph remains a valid RRNG, thus enabling efficient search with a single graph index. Based on this theoretical foundation, we propose a new graph index, called RNSG, as a practical solution that approximates RRNG efficiently. We develop fast algorithms for both
constructing the RNSG index and processing RFANN queries with it. Extensive experiments on five real-world datasets show that RNSG achieves significantly higher query performance with a more compact index and lower construction cost than existing state-of-the-art
methods.
\end{abstract}
\maketitle

\vspace*{-1em}
\pagestyle{\vldbpagestyle}
\begingroup\small\noindent\raggedright\textbf{PVLDB Reference Format:}\\
{\vldbauthors. \vldbtitle. PVLDB, \vldbvolume(\vldbissue): \vldbpages, \vldbyear.}\\
{\href{https://doi.org/\vldbdoi}{doi:\vldbdoi}}
\endgroup
\begingroup

\vspace*{-1em}
\renewcommand\thefootnote{}\footnote{\noindent
{This work is licensed under the Creative Commons BY-NC-ND 4.0 International License. Visit \url{https://creativecommons.org/licenses/by-nc-nd/4.0/} to view a copy of this license. For any use beyond those covered by this license, obtain permission by emailing \href{mailto:info@vldb.org}{info@vldb.org}. Copyright is held by the owner/author(s). Publication rights licensed to the VLDB Endowment. \\
\raggedright Proceedings of the VLDB Endowment, Vol. \vldbvolume, No. \vldbissue\ %
ISSN 2150-8097. \\
\href{https://doi.org/\vldbdoi}{doi:\vldbdoi} \\}
}\addtocounter{footnote}{-1}\endgroup

\vspace*{-1em}
\ifdefempty{\vldbavailabilityurl}{}{
\vspace{.3cm}
\begingroup\small\noindent\raggedright\textbf{PVLDB Artifact Availability:}\\
{\revRone{The source code, data, and/or other artifacts have been made available at \url{https://github.com/fallleaves01/R-NSG.git}.}}
\endgroup
}

\vspace*{-0.5em}
\section{Introduction}\label{sec:intro}

Nearest neighbor search in high-dimensional Euclidean spaces is a critical task for applications ranging from information retrieval~\cite{liu2007survey} and database systems~\cite{wang2021milvus,guo2022manu,wei2020analyticdb} to generative artificial intelligence~\cite{zhao2024retrieval} \yin{(RAG)}. However, the curse of dimensionality~\cite{beyer1999nearest,curse1,curse2} poses a fundamental challenge to traditional exact search methods~\yin{\cite{jagadish2005idistance}}, rendering them prohibitively slow for practical use. To overcome this, approximate nearest neighbor (ANN) search has emerged as a practical alternative, trading a marginal loss in accuracy for orders-of-magnitude gains in query efficiency.

Over the past decades, numerous ANN search techniques have been developed, including tree-based~\cite{beygelzimer2006cover,arora2018hd,ram2019revisiting}, hashing-based~\cite{LSHDatarIIM04,wang2014hashing,wang2017survey}, quantization-based~\cite{DBLP:conf/mm/TuncelFR02,gao2024rabitq,ge2013optimized,gong2012iterative,jegou2010product,aguerrebere2023similarity}, and graph-based methods~\cite{fu2017fast,fu2021high,harwood2016fanng,malkov2014approximate,malkov2018efficient,wang2021comprehensive}. Comprehensive benchmarks~\cite{wang2021comprehensive,25SIGMODgraphindexsurvey,liApproximateNearestNeighbor2020,dobson2023scaling} have established graph-based methods as the current state-of-the-art in terms of both query efficiency and accuracy. These techniques index data into a graph where nodes correspond to data objects, and edges are established based on \yin{spatial} proximity. A beam search algorithm~\cite{fu2017fast,wang2017survey,matsui2018survey} is then used for query processing. A key insight is that the top-performing graph indexes, including HNSW~\cite{malkov2018efficient} and NSG~\cite{fu2017fast}, are engineered as approximations of the \yin{Monotonic} Relative Neighborhood Graph (\yin{MRNG})~\cite{fu2017fast}. The \yin{MRNG} supports these methods with its desirable theoretical search guarantees and robust practical performance (see Section~\ref{sec:preliminary} for details).

Recent research has expanded to more complex ANN query variants to better serve real-world applications~\cite{zuo2024serf,gollapudi2023filtered,chronis2025filtered}.
A prominent example is the Range-Filtered Approximate Nearest Neighbor (RFANN) search, which has attracted significant interests~\cite{xu2024irangegraph,zuo2024serf,zhang2025efficient,liang2024unify}. In this setting, each data object comprises a vector and a totally-ordered numeric attribute. An RFANN query takes a query vector and a range constraint on the numeric attribute \yin{as input}, returning the nearest neighbors from the set of objects that satisfy the range filter. The RFANN query supports numerous practical applications. For example, on an e-Commerce platform~\yin{\cite{wei2020analyticdb}}, a user can search for products by a text description (encoded as a vector) while restricting results to a specific price range. Similarly, in Apple's on-device photo management system~\cite{pound2025micronn}, users can retrieve visually similar photos taken within a given date range, where both data images and query images are encoded into vectors to prevent personal information leakage, thereby meeting personalized needs efficiently while preserving data privacy.

Three basic strategies exist for RFANN queries~\cite{wang2021milvus}: pre-filtering, in-filtering, and post-filtering. These approaches share the characteristic of executing searches on a standard ANN index, differing primarily in when the range filter is applied. Specifically: (1) pre-filtering first selects in-range objects, then performs ANN search on this subset~\cite{wang2021milvus}; (2) in-filtering checks the range condition during the ANN search process~\cite{douze2024faiss}, skipping out-of-range objects; and (3) post-filtering conducts a standard ANN search first, then filters the results by range~\cite{li2018design}. The fundamental limitation of all these basic strategies is their inability to adapt efficiently to diverse query workloads with varying range selectivities. Although some hybrid methods attempt to dynamically combine these strategies~\cite{wang2021milvus}, they remain constrained by the same underlying limitations, inevitably yielding suboptimal performance.

Recent state-of-the-art methods directly integrate ANN graph indexes with numeric attributes~\cite{xu2024irangegraph,zuo2024serf,engels2024approximate,liang2024unify}, yet still struggle to balance efficiency and effectiveness. For instance, \cite{zuo2024serf} proposes building a dedicated ANN index for every possible query range. While this achieves high performance, it incurs prohibitive indexing costs. A subsequent compression technique \cite{zuo2024serf} reduces this overhead by constructing only multiple representative indexes, but at the cost of significantly degraded query accuracy. Similarly, \cite{xu2024irangegraph} integrates an ANN graph index with a \textit{segment tree} \cite{de2008computational}, building indexes for a moderate number of ranges; this improves practicality but still entails substantial indexing overhead. 

The fundamental difficulty in balancing index efficiency and query performance originates from a structural mismatch: while existing graph-based ANN indexes are designed based on spatial proximity, range filtering inherently disrupts these spatial relationships. For instance, \yin{MRNG}-inspired indexes organize edges according to spatial neighborhoods, but when a query range filters out intermediate nodes, the remaining graph often suffers from broken connectivity and poor searchability, causing significant performance degradation (see Fig.~\ref{fig:limitation-rng}).

To overcome these limitations, we propose a graph-structural model for RFANN queries: the Range-aware Relative Neighborhood Graph (RRNG). Unlike traditional graph indexes that consider only spatial proximity, the RRNG jointly considers spatial and attribute proximity. This additional attribute awareness leads to two key properties that we analyze in Section~\ref{sec:rrng}: (1) \textit{monotonic searchability}, extending the searchability intuition behind the \yin{MRNG}, and (2) \textit{structural heredity}, under which any range-induced subgraph remains an RRNG on the in-range objects. We further analyze the search and space overhead of RRNG relative to \yin{MRNG}. Taken together, these results suggest that RRNG provides a principled single-index model for RFANN search.

Similar to \yin{MRNG}, constructing the exact RRNG remains computationally expensive\yin{, with a time complexity of $O(n^3)$. }
To bridge theory and practice, we propose the Range-aware Neighborhood Search Graph (RNSG), a practical graph index that efficiently approximates the RRNG. The construction of RNSG proceeds in two stages: (1) building a \textit{local RRNG} for each node by considering both spatial and attribute proximity to form its candidate neighbor set, and (2) merging these local structures to approximate the global RRNG. We develop an efficient algorithm for this construction and prove that the resulting RNSG meets two crucial properties of strong connectivity and structural heredity, ensuring robust performance under any range constraint. Furthermore, we design a corresponding query processing algorithm enhanced by a novel entry node generation technique. In summary, our main contributions are as follows.

\noindent\textbf{Novel Theoretical Results.} We propose a novel RRNG theory for RFANN search. The RRNG not only preserves the monotonic searchability of traditional \yin{MRNG} but also achieves a new property of structural heredity--guaranteeing that any range-induced subgraph remains a valid RRNG, thereby overcoming a fundamental limitation of previous methods. We 
prove that RRNG achieves this with only a logarithmic overhead in both index size and search complexity compared to \yin{MRNG}. To our knowledge, RRNG is the first graph index that meets both monotonic searchability and structural heredity, establishing a theoretical foundation for future RFANN research.

\noindent\textbf{New Practical Algorithms.} Building upon the RRNG theory, we develop RNSG, a practical solution that efficiently approximates the RRNG. We present a fast range-aware pruning technique to accelerate RNSG construction, achieving a time complexity of $O(|C|M_{rnsg})$, where the parameter $|C|$ is a small constant in practice and $M_{rnsg}$ denotes the number of edges of RNSG, ensuring scalability for large datasets. We prove that RNSG satisfies the strong connectivity and structural heredity, guaranteeing robust performance under any query range. Leveraging this index, we also devise an efficient query processing algorithm equipped with a carefully-designed entry node generation technique.

\noindent\textbf{Extensive Experiments.} We conduct extensive experiments on five real-world datasets, evaluating our method against 8 competitive baselines. The results 
{show} that: (1) RNSG significantly outperforms 
existing state-of-the-art methods including iRangeGraph \cite{xu2024irangegraph} and UNIFY \cite{liang2024unify} in query performance. For instance, on the SIFT1M dataset, RNSG achieves 5645 QPS at 0.95 recall—representing a $2\times$ and $4\times$ improvement over iRangeGraph (2805 QPS) and UNIFY (1401 QPS) respectively; (2) RNSG maintains substantially more compact indexes, requiring several times less storage than competitors. For example, on the WIT dataset, RNSG maintains an index size of 0.38 GB, significantly smaller than iRangeGraph (1.5 GB) and UNIFY (8.9 GB). This represents a 3.9× to 23.4× reduction in memory overhead, demonstrating RNSG's exceptional space efficiency; (3) The index construction time of RNSG is substantially lower than iRangeGraph and UNIFY. For instance, on the WIT dataset, RNSG completes index construction in 1349 seconds, while iRangeGraph and UNIFY require 3495 and 11933 seconds respectively, representing a 2.59× to 8.85× speedup. Additional scalability tests confirm that RNSG maintains linear scaling with dataset size, underscoring its practical viability for large-scale applications. 

\vspace{-0.5em}
\section{Preliminaries} \label{sec:preliminary}
Let $D$ be a dataset comprising $n$ objects. Each object $o = (v, a)$ is composed of two features: a $d$-dimensional vector $o.v$, with $d$ typically being hundreds, and a numerical attribute $o.a \in \mathbb{R}$.

\noindent\textbf{Range-Filtered ANN Query (RFANN):} Given a dataset $D$, an RFANN query is parameterized by a tuple $q = \langle v, I, k \rangle$, where $q.v$ is a query vector, $q.I = [a_l, a_r]$ is a numeric range, and $q.k$ is the number of desired results. The query returns $q.k$ objects from $D$ that lie within the range $q.I$ and have the smallest Euclidean distance $\delta(q.v, o.v)$ to the query vector. 

Following prior work \cite{xu2024irangegraph}, we assume without loss of generality that the dataset is sorted in ascending order by the numeric attribute, i.e., $o_i.a \leq o_j.a$ for any $i \leq j$. \yin{For expository clarity, we assume distinct attribute values; duplicates are resolved by enforcing a deterministic ordering based on object IDs, which does not affect the correctness and performance of our method.} Given an RFANN query with range $[a_l, a_r]$, we can first use binary search to identify the smallest index $L$ and the largest index $R$ such that the subsequence $o_L, \dots, o_R$ contains precisely all objects satisfying $a_l \leq o_i.a \leq a_r$. The RFANN query then simplifies to finding the $k$ nearest neighbors to $q.v$ within this contiguous in-range subsequence. 

\noindent\textbf{Graph-based Indexes for ANNS.} 
Graph-based methods construct a proximity graph $G = (V, E)$ during the indexing phase, where each node $x \in V$ corresponds to a data object $o \in D$. At query time, the beam search algorithm is typically employed~\cite{wang2017survey}. Specifically, the beam search starts from an entry node (often the graph's approximate center) and maintains two dynamic sets: a candidate set $\mathcal{S}$ (a min-heap) and a result set $\mathcal{R}$ (a max-heap of size $k$). $\mathcal{R}$ stores the current $k$ nearest neighbors, while $\mathcal{S}$ contains promising candidates for further exploration.

In each iteration, the algorithm expands the candidate with the smallest distance in $\mathcal{S}$. Its unvisited neighbors are then evaluated and added to $\mathcal{S}$. A neighbor is inserted into $\mathcal{R}$ if it is closer to the query than the farthest object currently in $\mathcal{R}$. To balance accuracy and efficiency, a hyperparameter $ef_{\text{search}}$ limits the size of $\mathcal{S}$; when exceeded, the farthest candidate is removed. The search terminates when the minimum distance in $\mathcal{S}$ exceeds the maximum distance in $\mathcal{R}$, at which point $\mathcal{R}$ is returned as the final result.

Most state-of-the-art graph-based methods, including HNSW~\cite{malkov2018efficient} and NSG~\cite{fu2017fast}, are inspired by the {Monotonic} Relative Neighborhood Graph ({MRNG}) introduced in~\cite{fu2017fast}, as defined below: 

\begin{definition}[\revRone{{Monotonic Relative Neighborhood Graph~\cite{fu2017fast}}}] \label{def:rng}
\revRone{
Given a finite dataset \(D\) in a metric space with distance function \(\delta\), define the lune of two distinct objects \(x,y\in D\) as:}
\[
\operatorname{lune}(x,y)=
\{z\in D\setminus\{x,y\}\mid
\delta(x,z)<\delta(x,y),\ \delta(y,z)<\delta(x,y)\}.
\]
\revRone{The \yin{MRNG} is a directed graph \(G=(V,E)\) whose outgoing edges satisfy the recursive condition}
\[
(x\to y) \in E
\Longleftrightarrow
\operatorname{lune}(x,y)=\emptyset
\ \text{or}\
\forall z\in\operatorname{lune}(x,y),\ (x\to z)\notin E .
\]
\end{definition}

\begin{figure}[!t]
\vspace*{-3em}
\begin{center}
\subcaptionbox{\yin{MRNG} for ANN query\label{fig:rng}}{
\includegraphics[width=0.45\columnwidth]{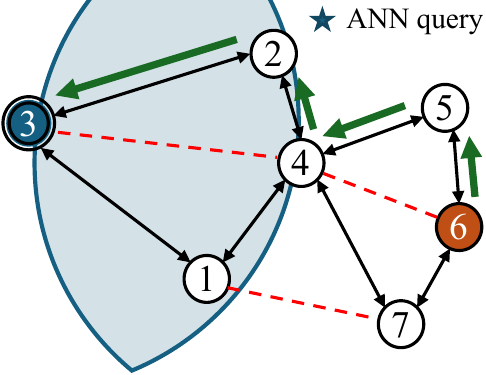}
}
\subcaptionbox{\yin{MRNG} for RFANN query\label{fig:rng-range}}{
\includegraphics[width=0.45\columnwidth]{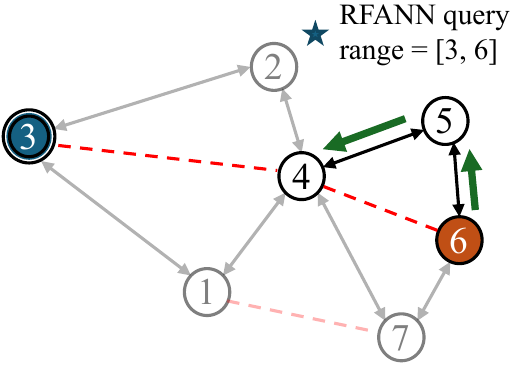}
}
\vspace{-0.5em}
\caption{Illustration of \yin{MRNG} for ANN and RFANN queries (red dashed edges are pruned by \yin{MRNG}), node labels denote numeric attribute values. (a) Under a standard ANN query, the \yin{MRNG} correctly preserves a path to the true nearest neighbor (node 3). (b) Under an RFANN query with a range constraint, the required path is broken, making the true nearest neighbor unreachable from the query node.}

\label{fig:limitation-rng}
 \vspace*{-1.0em}
\end{center}
\end{figure}

\revRone{
As illustrated in Fig.~\ref{fig:limitation-rng}, the direct edge \(4\to 3\) is pruned because object \(2\) lies in the lune of \(4\) and \(3\), while the edge \(4\to 2\) has already been selected.
This pruning strategy is desirable for ANN queries, as it removes unnecessary outgoing edges while preserving connectivity  (e.g., \(4\to 2\to 3\)). Now we provide its formal search-quality and time-complexity guarantees~\cite{fu2017fast}.
}




\begin{lemma}[\cite{fu2017fast}]
\label{lemma:rng-beam}
Given a dataset $D$ of $n$ points in a metric space with distance function $\delta$, let $G(V, E)$ be an \yin{MRNG} constructed on $D$. For any node $x \in V$, its nearest neighbor $y \in V$ can be found by performing a beam search on $G$.
\end{lemma}

\begin{lemma}[\cite{fu2017fast}]\label{lem:MRNG-complex}
The time complexity of performing the beam search on the \yin{MRNG} is $O(n^{1/d} \log n) \approx O(\log n)$.
\end{lemma}


However, constructing an exact \yin{MRNG} for a dataset of size $n$ has a time complexity of $O(n^3)$~\cite{fu2017fast}, which is prohibitively expensive. Consequently, practical methods (e.g., HNSW and NSG) construct graph indexes by approximating the \yin{MRNG}. Typically, they first retrieve hundreds of candidate neighbors (e.g., via approximate nearest neighbor search) and then apply an MRNG-style ordered pruning rule to this candidate set. A hyperparameter $m$ limits the maximum out-degree of each node, ensuring sparsity by retaining only the top-$m$ neighbors after pruning~\cite{wang2021comprehensive}.

\revRone{Despite the advantages of MRNG for ANN queries, it faces challenges in RFANN settings: although the MRNG is navigable before filtering, the routing path may rely on vertices that are removed during the RFANN query. For example, under the query range \([3,6]\), node \(2\) is filtered out, so a route from the entry node \(6\) to the in-range target node \(3\) may disappear even though both endpoints remain in the graph. Thus, the issue is that MRNG lacks monotonicity in the filtered graph. To address this, various RFANN methods have been proposed, as discussed in the Introduction.}

\revRone{
\noindent\textbf{Limitations of Existing RFANN Search Methods.} Existing RFANN methods can be broadly grouped into filter-placement strategies and range-aware hybrid indexes~\cite{wang2021milvus,engels2024approximate,liang2024unify,zuo2024serf,xu2024irangegraph}. The first group applies the range filter before, during, or after search~\cite{wang2021milvus,douze2024faiss,li2018design}. These strategies are easy to deploy, but none of them 
performs well across all selectivities: pre-filtering becomes expensive on large ranges, post-filtering wastes work on small ranges, and in-filtering may damage graph navigability. The second group incorporates attribute information more directly into the index by materializing range-aware ANN indexes such as 
SeRF~\cite{zuo2024serf} and iRangeGraph~\cite{xu2024irangegraph}. These methods improve specific workloads, but 
typically 
introduces index redundancy. 
This motivates a single graph model whose structural properties are preserved under arbitrary query ranges.
}

\section{Range-aware Relative Neighborhood Graph} \label{sec:rrng}
\revRone{The Range-aware Relative Neighborhood Graph (RRNG) is the ideal graph model that motivates RNSG. RRNG extends the formal \yin{MRNG} witness condition by requiring a pruning witness to be relevant in both vector space and attribute order.} \revRone{This change makes the model suitable for range filters: if both endpoints of an edge remain inside an interval, then every valid attribute-between witness for that edge also remains inside the same interval. We first define RRNG and then prove the two properties used later by RNSG: monotonic searchability and heredity under interval restriction.}

\begin{figure}[!t]
\vspace*{-3em}
  \centering
  \begin{subfigure}[t]{0.45\linewidth}
    \centering
    \includegraphics[page=1,width=\linewidth]{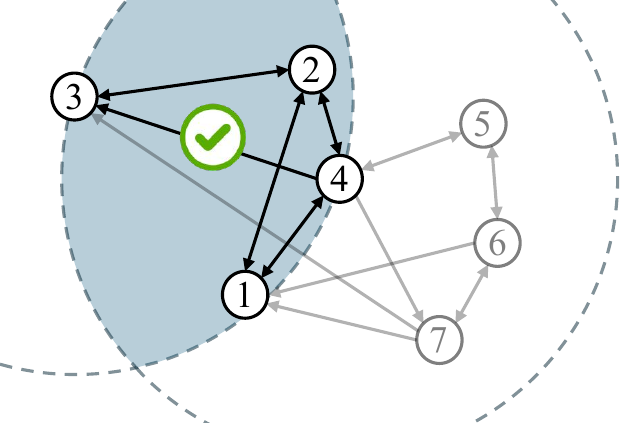}
    \caption{Edges retained by RRNG}
  \end{subfigure}
  \hfill
  \begin{subfigure}[t]{0.45\linewidth}
    \centering
    \includegraphics[page=1,width=\linewidth]{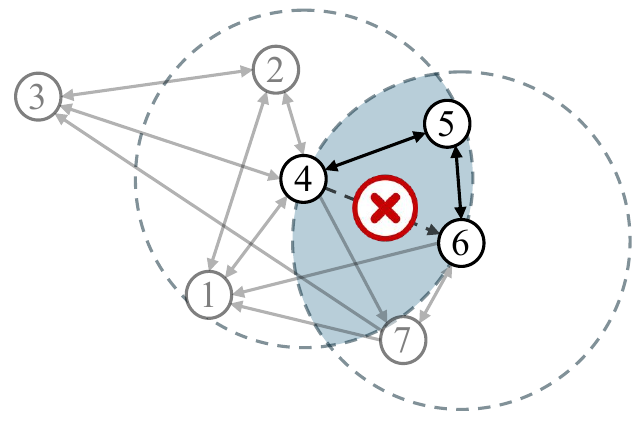}
    \caption{Edges pruned by RRNG}
  \end{subfigure}
  \vspace{-0.4em}
  \caption{Illustration of the RRNG pruning strategy, node labels denote numeric attribute values.}
  \label{fig:rrng-pruning-examples}
\end{figure}

\begin{definition}[RRNG]\label{def:rrng}
Given a dataset \(D\) in a metric space with distance function \(\delta\), for a source node \(x\) and a candidate \(y\), define the range-aware lune
\[
\begin{aligned}
\operatorname{rlune}(x,y)=
\{z\!\in\!D\!\setminus\!\{x,y\}\mid\;&
\delta(x,z)<\delta(x,y),\\
&\delta(y,z)<\delta(x,y),\\
&z.a\in(\min\{x.a,y.a\},\max\{x.a,y.a\})\}.
\end{aligned}
\]
The RRNG is a directed graph \(G=(V,E)\) whose outgoing edges satisfy the recursive condition
\[
x\to y\in E
\Longleftrightarrow
\operatorname{rlune}(x,y)=\emptyset
\ \text{or}\
\forall z\in\operatorname{rlune}(x,y),\ x\to z\notin E .
\]
\end{definition}

\revRone{As in \yin{MRNG}, this definition is recursive}. For a fixed source node \(x\), candidates on each side of \(x.a\) can be processed in nondecreasing attribute gap \(|y.a-x.a|\), with ties broken deterministically. If \(z\) is a valid range-aware witness for \(x\to y\), then \(z.a\) lies between \(x.a\) and \(y.a\), so \(z\) is processed before \(y\) in the same-side scan. Therefore the recursive definition is realized by retaining \(x\to y\) exactly when no previously retained outgoing neighbor of \(x\) lies in \(\operatorname{rlune}(x,y)\). The graph may contain both directions for a pair of nodes, but \revRone{the two directed edges are selected by the pruning rule of their respective source nodes.} Figures use simple edge drawings when the direction is not material; \revRone{search paths and pruning statements below use the directed outgoing edges.}


\begin{example} \label{example:rrng}
    Fig.~\ref{fig:rrng-pruning-examples} illustrates the RRNG structure, with node numbers indicating numeric attributes. In subfigure (a), the candidate edge is $(3,4)$. Although node $2$ lies in the geometric witness region, it does not satisfy the attribute-between condition because $2 \notin (3,4)$, so RRNG retains edge $(3,4)$. In subfigure (b), the candidate edge is $(4,6)$, and node $5$ satisfies both the geometric witness condition and the attribute-between condition, so RRNG prunes edge $(4,6)$.   
\end{example}

\noindent\textbf{Differences between \yin{MRNG} and RRNG.} Although Example~\ref{example:rrng} 
\yin{illustrates} the RRNG pruning rule, it would be misleading to view RRNG as simply an \yin{MRNG} with fewer pruning opportunities. RRNG changes which vertices are allowed to serve as witnesses, so the retained edge set can change in both directions. As shown in Fig.~\ref{fig:rng_to_rrng}, node labels again denote attribute values. The upper sketches visualize the generic replacement pattern, while the lower pair gives a concrete four-node example. In the \yin{MRNG} in subfigure (a), edge $(1,4)$ is preserved as a direct monotonic connection. In the corresponding RRNG in subfigure (b), edge $(1,4)$ is pruned because node $2$ lies in the attribute interval $(1,4)$ and witnesses triangle $(1,2,4)$. At the same time, RRNG retains edges such as $(1,2)$ and $(2,4)$, so neither graph's edge set contains the other's.

In summary, \yin{MRNG} establishes edges based solely on \textit{spatial} proximity, whereas RRNG incorporates both \textit{spatial} and \textit{attribute} proximity into the witness rule. This design principle leads to distinct graph topologies: direct edges connecting distant attributes (e.g., $(1,4)$ in Fig.~\ref{fig:rtrr_rng}) in \yin{MRNG} are often replaced in RRNG by paths traversing attribute-between nodes (e.g., $(1,2)$ and $(2,4)$ in Fig.~\ref{fig:rtrr_rrng}). This replacement makes the path more robust to interval filtering.

\begin{figure}[!t]
\vspace*{-3em}
\begin{center}
\subcaptionbox{\yin{MRNG}\label{fig:rtrr_rng}}{
\includegraphics[width=0.45\columnwidth]{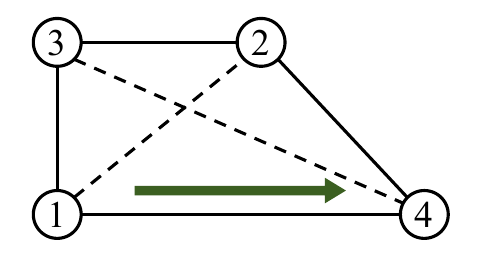}
}
\subcaptionbox{RRNG\label{fig:rtrr_rrng}}{
\includegraphics[width=0.45\columnwidth]{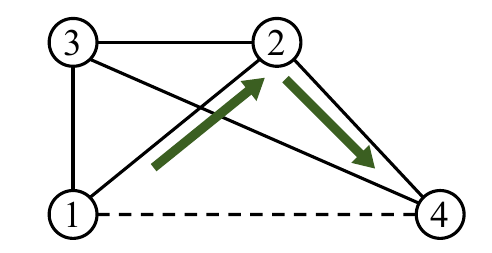}
}
\vspace{-0.4em}
\caption{Differences between \yin{MRNG} and RRNG, node labels denote numeric attribute values. The RRNG tends to replace edges (e.g., the edge $(1,4)$) with large attribute differences with those having smaller attribute differences (e.g., the edges $(1,2)$ and $(2,4)$). In the \yin{MRNG}, $1\rightarrow 4$ is a monotonic path, while in the RRNG $1\rightarrow 2 \rightarrow 4$ is a monotonic path. }
\vspace{-0.4em}
\label{fig:rng_to_rrng}
\end{center}
\end{figure}

\subsection{Theoretical Properties of RRNG} \label{subsec:rrng-theory}
\revRone{This subsection establishes the two RRNG properties that address the main failure mode of applying ordinary ANN graphs to RFANN. Monotonic searchability gives the same kind of target-reaching guarantee used by \yin{MRNG}-based ANN search, while heredity ensures that this property is not destroyed by taking an interval-induced subgraph. We begin by defining monotonic paths.}
\begin{definition}[Monotonic Path]
    \label{def:mono-path}
    Given a graph $G(V, E)$ and a distance function \( \delta \),  a directed path from node $x\in V$ to $y \in V$, denoted by $ x = v_0 \to v_1 \to \cdots \to v_m = y$, is referred to as a monotonic path if and only if
$\delta(v_{i+1}, y) < \delta(v_i, y), \quad \forall i \in \{0, 1, \dots, m-1\}.$
\end{definition}

\revRone{The following theorem shows that for any two nodes $x$ and $y$ in RRNG, there exists a path from $x$ to $y$ whose distance to the target $y$ strictly decreases at every step.}
\begin{theorem}[Monotonic Searchability Property]\label{thm:rrng-monotone}
Given a dataset \( D \) in a metric space with distance function \( \delta \), let \( G = (V, E) \) be an RRNG constructed over \( D \) as defined in Definition~\ref{def:rrng}. Then, for any two distinct nodes \( x, y \in V \), there exists a monotonic path 
 from $x$ to $y$ in $G$, along which the distance to $y$ strictly decreases at every step.
\end{theorem}

\begin{proof}
\revRoneStart
Fix the target node \(y \in V\). Let \(u_1,u_2,\dots,u_t,\dots\) denote all nodes \(u\neq y\), sorted by increasing distance to \(y\):

\[
  \delta(u_1,y) < \delta(u_2,y) < \cdots.
\]
We prove by induction on this order that every \(u_t\) has a directed monotonic path to \(y\).

\noindent\textbf{Base step:}
Consider \(u_1\), the node closest to \(y\). The directed edge \(u_1\to y\) cannot be pruned in RRNG, because any valid outgoing witness \(z\) for pruning it would have to satisfy \(\delta(z,y) < \delta(u_1,y)\), contradicting the choice of \(u_1\). Hence \(u_1 \to y\) is a monotonic path.

\noindent\textbf{Inductive step:}
Assume that every node \(u_j\) with \(\delta(u_j,y) < \delta(u_t,y)\) already has a monotonic path to \(y\). We now consider \(u_t\).

\textbf{Case 1:} If \(u_t\to y\in E\), then the single edge \(u_t \to y\) is already a monotonic path.

\textbf{Case 2:} If \(u_t\to y\notin E\), then by Definition~\ref{def:rrng} there exists an already retained outgoing witness \(z\) of \(u_t\) such that
\[
\delta(u_t,z) < \delta(u_t,y),\qquad
\delta(z,y) < \delta(u_t,y),
\]
and \(u_t\to z\in E\).
The inequality \(\delta(z,y) < \delta(u_t,y)\) means that \(z\) appears earlier than \(u_t\) in the above order. By the inductive hypothesis, \(z\) has a monotonic path
\[
z = v_0 \to v_1 \to \cdots \to v_m = y
\]
such that \(\delta(v_{i+1},y) < \delta(v_i,y)\) for every \(i\). Since \(\delta(z,y) < \delta(u_t,y)\), prepending the directed edge \(u_t\to z\) gives
\[
u_t \to z \to v_1 \to \cdots \to y,
\]
whose distance to \(y\) also strictly decreases at the first step. Therefore this concatenated path is monotonic.

By induction, every node has a directed monotonic path to \(y\). Since \(y\) was arbitrary, the theorem holds for every ordered pair of distinct nodes \(x,y\in V\).
\reviseend
\end{proof}

\revRone{The induction proof above can be read directly from Fig.~\ref{fig:rtrr_rrng}. In the lower RRNG example, the direct outgoing edge \(1\to4\) is absent because node \(2\) witnesses triangle \((1,2,4)\). However, node \(2\) is strictly closer to \(4\) than node \(1\) is, and edge \(1\to2\) exists. The proof therefore recurses on the shorter pair \((2,4)\), yielding the monotonic path \(1 \rightarrow 2 \rightarrow 4\). Note that the \yin{MRNG} also satisfies a monotonic-search property~\cite{fu2017fast}; Fig.~\ref{fig:rtrr_rng} illustrates the direct monotonic path \(1\rightarrow4\) in that case.}

\begin{corollary}\label{lem:efficient_search}
Let \( G = (V, E) \) be an RRNG constructed from a dataset \( D \) in a metric space with distance function \( \delta \). For any target object \(y\in V\), best-first graph search on \(G\), using \(\delta(\cdot,y.v)\) as the distance key and a search list large enough not to discard the current closest candidate, reaches \(y\) regardless of the entry node \(x\).
\end{corollary}
This follows directly from Theorem~\ref{thm:rrng-monotone}: best-first expansion can repeatedly advance along a distance-decreasing outgoing neighbor until it reaches the target. The full argument is given in \rnsgappendixref{app:proof-efficient-search}.

RRNG also satisfies a crucial \textit{hereditary} property that enables efficient RFANN processing.

\begin{theorem}[Hereditary Property]\label{thm:rrng-closure}
Let \( G = (V, E) \) be an RRNG constructed from a dataset \( D \) in a metric space with distance function \( \delta \). For any RFANN query interval \(I=[a_l,a_r]\), let \(V_I=\{x\in V\mid x.a\in I\}\) and let \(G[I]=(V_I,E')\) be the directed subgraph induced by \(V_I\), where \(E'=\{(u,v)\in E\mid u,v\in V_I\}\). Then \(G[I]\) is exactly the RRNG constructed on \(V_I\).
\end{theorem}
\begin{proof}
Let \(H[I]=(V_I,E_I)\) be the RRNG constructed directly on \(V_I\). Fix a source node \(x\in V_I\). We show that the outgoing neighbors of \(x\) in \(G\), restricted to \(V_I\), are exactly the outgoing neighbors selected for \(x\) in \(H[I]\). Consider candidates on the right side of \(x.a\); the left side is symmetric. If \(y\in V_I\) and \(y.a>x.a\), then every right-side candidate processed before \(y\) has attribute in \((x.a,y.a)\). Since \(x\) and \(y\) both lie in the interval \(I\), every such earlier candidate also lies in \(V_I\).

Therefore the prefix of candidates that can affect the prune-or-keep decision for \(y\) is identical in the full RRNG construction and in the construction restricted to \(V_I\). By induction over this candidate order, the retained outgoing-witness set before processing \(y\) is also identical in the two constructions. The geometric witness test is unchanged, and any valid attribute-between witness for \(x\to y\) lies inside \(I\). Hence \(x\to y\) receives the same decision in \(G\) and in \(H[I]\). Applying the same argument to the left side and then to every \(x\in V_I\) proves \(E'=E_I\).
\end{proof}

\revRone{By Theorem~\ref{thm:rrng-closure}, the subgraph induced from the RRNG by any query range remains a \textit{correct} RRNG, preserving the theoretical searchability properties of the ideal graph. Consequently, the RRNG model does not require a separate graph for every query interval. This is the structural property that motivates a single-index design for RFANN, in contrast to methods that materialize multiple range-specific graph structures~\cite{engels2024approximate,zuo2024serf,xu2024irangegraph}.}

\revRone{This theorem also makes the theoretical advantage of RRNG over MRNG under range filtering explicit. In MRNG-style graphs, a path can rely on an intermediate node that is later removed by the query range, so the filtered subgraph may lose navigability. In RRNG, by contrast, pruning witnesses must themselves lie between the two endpoints in attribute order. Therefore the same attribute interval that keeps the endpoints also keeps every valid witness, which is exactly why the RRNG construction commutes with interval restriction.}

\begin{example}
    \revRone{Fig.~\ref{fig:rrng-induce-demo} illustrates the hereditary property described in Theorem~\ref{thm:rrng-closure}. Given an RFANN query with range [3,6], nodes 1 and 2 are filtered out, but the induced subgraph still preserves an in-range search route 6-5-4-3. The key reason is the attribute-between witness rule: if an edge's two endpoints are kept by an interval, then any valid witness that could justify pruning that edge is also inside the same interval. Thus the induced subgraph is a valid RRNG on the in-range objects.}
\end{example}

\begin{figure}[!t]
\vspace*{-3em}
  \centering
  \includegraphics[width=0.9\linewidth]{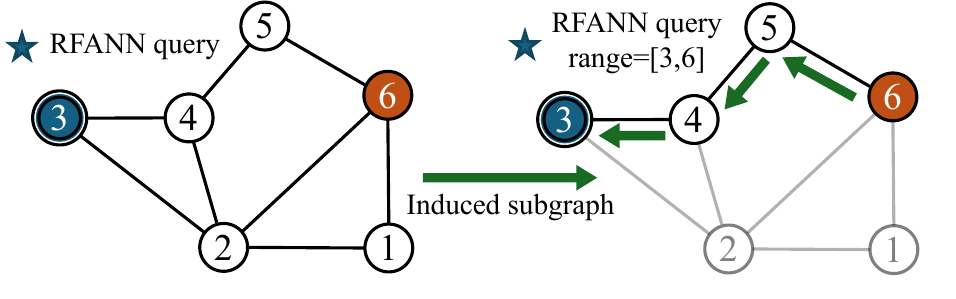}
    \vspace{-0.3em}
    \caption{Illustration of the hereditary property of RRNG, \revRone{node labels denote numeric attribute values}. The left subfigure shows the complete RRNG graph, while the right subfigure depicts the subgraph induced by a query range $[3, 6]$. }
    \vspace{-0.4em}
  \label{fig:rrng-induce-demo}
\end{figure}

\subsection{Complexity Analysis} \label{subsec:rrng-complexity}
This subsection presents complexity analyses of the RRNG. We begin by examining the expected cost of node-to-node monotone search on the RRNG, then proceed to analyze the RRNG size, and finally establish the time complexity of its construction. Due to space limitations, all missing proofs are provided in the full version~\cite{fullversion}.

\begin{theorem}\label{thm:search-complex}
The expected cost of node-to-node monotone search on an RRNG is $O(n^{1/d} \log^2 n) \approx O(\log^2 n)$.
\end{theorem}

Theorem~\ref{thm:search-complex} shows that the ideal RRNG retains the logarithmic-style monotone-search behavior of \yin{MRNG}-based ANN graphs, with an additional logarithmic factor from preserving attribute-order witnesses. Together with Theorem~\ref{thm:rrng-closure}, this property is preserved inside arbitrary query intervals.

\begin{theorem}\label{thm:space-complex}
The size of an RRNG is bounded by $O(S_r\log n)$, where $S_r$ is the expected size of \yin{MRNG}.
\end{theorem}

Theorem~\ref{thm:space-complex} shows the RRNG space overhead is $O(\log n)$ relative to the \yin{MRNG}. More importantly for RFANN indexing, the single-graph design avoids the multiplicative storage overhead incurred by methods that precompute or compose multiple graph structures across ranges \cite{zuo2024serf,xu2024irangegraph,engels2024approximate}.

\noindent\textbf{Basic Approach to Construct RRNG.} Definition~\ref{def:rrng} suggests a direct exact construction. For each source node \(x\), we split all other nodes into the two attribute sides of \(x.a\), process each side by nondecreasing attribute gap, and decide whether to keep \(x\to y\) by checking earlier same-side candidates \(z\) as possible outgoing witnesses. A literal implementation enumerates source-candidate-witness triples \((x,y,z)\), which gives the cubic bound below.

\begin{theorem} \label{thm:rrng-built-time}
The time complexity of the basic RRNG construction algorithm is $O(n^3)$.
\end{theorem}
The proof follows by enumerating, for every source-candidate pair, all earlier same-side third-node witnesses; details appear in \rnsgappendixref{app:proof-rrng-build-time}.

By Theorem~\ref{thm:rrng-built-time}, constructing an RRNG graph requires ${O}(n^3)$ time, the same as that of the \yin{MRNG} \cite{arya1993approximate}, which makes it impractical. The next section therefore constructs an approximate RRNG graph index, the range-aware neighborhood search graph (RNSG).


\section{The Proposed RNSG Approach} \label{sec:rnsg}
\revRone{Exact RRNG construction is useful as a theoretical target but is too expensive for large datasets. We therefore propose the range-aware neighborhood search graph (RNSG), a practical approximation of RRNG. RNSG does not enumerate all possible witnesses. Instead, for each node it forms a bounded candidate set from two sources: spatial neighbors from an approximate KNN graph and nearby objects in attribute order. It then applies the RRNG pruning rule locally to this candidate set. The resulting graph is a single index for all query intervals. The following subsections first describe the pruning rule, then define the bounded candidate source used by practical RNSG, state the interval-restriction property that remains valid under this fixed source, and finally give the query algorithm.}

\begin{algorithm}[t]
\small
\setcounter{AlgoLine}{0}
\LinesNumbered
\caption{Fast Range-Aware Pruning Strategy}
\label{alg:RAP}
\KwIn{A node $x$, a candidate set $C$ for $x$, maximum edges $m$}
\KwOut{Neighbor set $NS$ of $x$}

\SetKwFunction{RRNGPrune}{RRNGPrune}
\SetKwFunction{Prune}{Prune}
\SetKwProg{Proc}{Procedure}{}{}

\Proc{\RRNGPrune{$x, C, m$}}{
    $NS \leftarrow \varnothing$\;
    $C_l \gets \{\, v \in C \mid v.a < x.a \,\}$\;
    $C_r \gets \{\, v \in C \mid v.a > x.a \,\}$\;
    Sort $C_l$ and $C_r$ in ascending order of $|v.a - x.a|$\;
    $R_l \leftarrow$ \Prune($x$, $C_l$, $\frac{m}{2}$)\;
    $R_r \leftarrow$ \Prune($x$, $C_r$, $\frac{m}{2}$)\;
    $NS \leftarrow R_l \cup R_r$\;
    \Return $NS$ \;
}
\Proc{\Prune{$x, C, m$}}{
    $R \leftarrow \varnothing$\;
    \For{$v_i \in C$}{
        keep $\leftarrow \text{True}$\;
        \For{$v_j \in R$}{
        \If{$\delta(x,v_j) < \delta(x,v_i)$ \textbf{and}\\
        \hspace{2em}$\delta(v_i,v_j) < \delta(x,v_i)$}{
            keep $\leftarrow$ \text{False}\;
            \textbf{break}\;
            }            
        }
        \If{\textnormal{keep} \textbf{and} $|R|<m$}{
            $R \leftarrow R\cup v_i$\;
        }
    }
    \Return $R$\;
}
\end{algorithm}

\subsection{Fast Range-Aware Pruning Strategy} \label{subsec:fast-rrng-prune}
\revRone{Algorithm~\ref{alg:RAP} is the local pruning routine used by RNSG. Given a center node $x$ and its candidate set $C$, the algorithm first separates candidates on the left and right sides of $x$ in attribute order. This separation is not an implementation trick; it follows from the witness condition in Definition~\ref{def:rrng}. A candidate on one side of $x.a$ cannot be an attribute-between witness for a candidate on the other side, so the two sides can be pruned independently and then merged.}
\begin{lemma}
\label{lemma:prune-divide}
For a center node $x$ with candidate set $C$, partition $C$ into two sets $C_l$ and $C_r$ where $C_l=\{y\in C \mid y.a < x.a\}$ and $C_r=\{y\in C \mid y.a > x.a\}$. Then this partition ensures independence in RRNG pruning: edges from $x$ to any node in $C_l$ are unaffected by nodes in $C_r$, and vice versa.
\end{lemma}
\revRone{The proof is a direct consequence of attribute betweenness and is deferred to \rnsgappendixref{app:proof-prune-divide}.}

\revRone{Lemma~\ref{lemma:prune-divide} is not itself a pruning criterion; rather, it is a \emph{decomposition} result. It tells us that candidates on opposite sides of \(x.a\) can never witness each other, so Algorithm~\ref{alg:RAP} may safely solve two smaller one-sided pruning problems and then merge their outputs. For example, if \(x.a=4\) and a left-side candidate has attribute \(3\), then any right-side candidate with attribute \(5\) or \(6\) cannot prune edge \((x,3)\), because a valid witness would have to lie in the interval \((3,4)\).}

\revRone{In implementation, there is no need for an additional expensive pass devoted only to partitioning. While collecting candidates, we append each candidate directly to the left or right buffer, and then process each side in increasing attribute gap. We present partitioning and ordering as separate logical steps in Algorithm~\ref{alg:RAP} to make the correctness argument explicit. The linear partitioning cost is dominated by the subsequent pruning scan and does not change asymptotic complexity.}

For a node $x$, after partitioning the candidate set $C$ into $C_l$ and $C_r$, we sort each subset by the absolute difference between the candidate's attribute and $x.a$ (line~5). This ordering is crucial because Lemma~\ref{lemma:prune-order} guarantees that a candidate with a smaller attribute gap cannot be pruned by another candidate with a larger gap.
\begin{lemma}
\label{lemma:prune-order}
Let $x$ be a node and $C_l=\{v_i\in C|\ v_i.a<x.a\}$ be sorted by $|x.a-v_i.a|$ in ascending order. For any $c_i,c_j\in C_l$ with $|c_i.a-x.a|<|c_j.a-x.a|$, the edge $(x,c_i)$ cannot be pruned by $(x,c_j)$.
The case for $C_r=\{v_i\in C|\ v_i.a>x.a\}$ is symmetric.
\end{lemma}
\revRone{The proof again follows from the attribute-between condition: a larger-gap candidate lies outside the interval needed to witness a smaller-gap candidate. \rnsgappendixref{app:proof-prune-order} gives the formal argument.}

\revRone{Lemma~\ref{lemma:prune-order} enables us to process nodes in ascending order of their attribute gap to $x$ (line 5). }

\revRone{For example, suppose \(x.a=4\) and two right-side candidates have attributes \(5\) and \(6\).The candidate with attribute \(6\) cannot prune the edge from \(x\) to the candidate with attribute \(5\), because a valid range-aware witness for this edge must lie between attributes \(4\) and \(5\).However, the candidate with attribute \(5\) may prune the edge from \(x\) to the candidate with attribute \(6\), since it satisfies the attribute-between condition.

Thus, Lemma~\ref{lemma:prune-order} is used to justify the ordered one-pass scan: later larger-gap candidates never invalidate decisions made for earlier smaller-gap candidates.}
This ordered processing ensures no valid edges are missed. We then perform sequential pruning within each subset, initializing a result set $R$ and inserting nodes in the sorted order (line 12). A node is pruned if it forms the longest edge with any previously added neighbor (lines 13-20). Finally, we select at most $m/2$ edges from each subset (lines 6-7), merge them (line 8), and return the final neighbor set (line 9). Here, $m$ is a hyper-parameter used to control the size of the neighborhood. Note that if we set the candidate set $C$ to the entire dataset $D$ and set $m = \infty$, then we can use Algorithm~\ref{alg:RAP} to construct the exact RRNG, indicating the correctness of our range-aware pruning strategy.

\begin{theorem}\label{thm:RAP-gets-RRNG}
For a dataset $D$, if Algorithm~\ref{alg:RAP} is applied to every object $x \in D$ with parameters $C = D$ and $m = \infty$, the resulting graph, comprising all objects and their computed neighborhood sets $NS(x)$, precisely yields the RRNG defined in Definition~\ref{def:rrng}.
\end{theorem}

\revRone{This theorem also marks the boundary between the ideal and practical graphs. With the full dataset as the candidate set and no degree cap, RAP exactly recovers RRNG. Practical RNSG starts when this full candidate set is replaced by a bounded candidate source and the retained degree is capped; the resulting graph is therefore an approximation of RRNG, not an edge-exact implementation of the ideal graph.}

Below, we show that the time complexity of RRNG construction can be reduced from $O(n^3)$ to $O(n M_{rrng} )$ ($M_{rrng}< n^2$) by employing our fast range-aware pruning strategy, where $M_{rrng}$ is the number of RRNG edges.

\begin{theorem}\label{thm:prunning-complexity}
The time complexity of constructing an exact RRNG using Algorithm~\ref{alg:RAP} is $O(nM_{rrng} )$.
\end{theorem}
The proof sums the per-node pruning cost over the final RRNG degrees and is provided in \rnsgappendixref{app:proof-pruning-complexity}. 

Although our fast range-aware pruning strategy reduces the time complexity of RRNG construction to $O(nM_{rrng} )$, this cost remains prohibitive for large-scale applications. To further improve the efficiency, we present RNSG, a practical graph index that leverages our range-aware pruning strategy to efficiently approximate the RRNG.
\begin{algorithm}[t]
\setcounter{AlgoLine}{0}
\LinesNumbered
\caption{The RNSG Construction Algorithm}
\label{alg:rnsg}
\small
\KwIn{
Dataset $D$, spatial candidate size $ef_\mathrm{spatial}$, attribute window size $ef_\mathrm{attribute}\ge 1$, and maximum edges $m$}
\KwOut{A RNSG $G=(V,E)$}

\SetKwFunction{RNSGConstruct}{RNSGConstruct}
\SetKwProg{Fn}{Function}{}{}

\Fn{\RNSGConstruct{$D, ef_\mathrm{spatial}, ef_\mathrm{attribute}, m$}}{
    Obtain an approximate $ef_\mathrm{spatial}$-NN candidate graph $G_s=(V,E_s)$ from $D$\;
Let $\{v_1, \cdots, v_n\}$ be sorted in ascending order by the nodes' numeric attributes\;
    $E \leftarrow \varnothing$\;

    \For{$i \leftarrow 1$ \KwTo $n$}{
        $C \leftarrow \mathrm{Neighbour}_{G_s}(v_i)$\;
        \For{$j \in [\max(1,i - ef_\mathrm{attribute}),\, \min(n,i + ef_\mathrm{attribute})]\setminus\{i\}$}{
            $C.\mathrm{add}(v_j)$\;
        }

        $N_i \leftarrow$ \RRNGPrune{$v_i, C, m$}\;
        \ForEach{$x \in N_i$}{
            $E \leftarrow E \cup \{(v_i, x)\}$\;
        }
    }

    \Return $G=(V,E)$\;
}
\end{algorithm}
\subsection{The RNSG Construction Algorithm} \label{subsec:rnsg-construction-alg}
\revRoneStart
Algorithm~\ref{alg:rnsg} details the construction of the RNSG index. The process begins by building a candidate neighborhood for each node $x$ that captures both spatial and attribute proximity. Spatially, we employ the classic NNDescent algorithm~\cite{2011nn-descent} to construct an approximate $ef_\mathrm{spatial}$-NN candidate graph (line~2), which identifies neighbors based on the distance function $\delta$. This graph is used only as a candidate source: its raw edges are not kept in the final RNSG unless they survive range-aware pruning. For a fixed dataset, distance metric, and KNNG degree, this candidate graph is a dataset-level spatial structure rather than an RNSG parameter choice. It can therefore be built once and reused to construct multiple RNSG variants with different range-window or pruning parameters, and it may also serve other graph-index construction pipelines that consume kNN candidates. After the final RNSG is produced, the temporary KNNG can be discarded; the deployed RNSG index stores only the retained RNSG graph and its query-time auxiliary structures.
For attribute proximity, we do not use an attribute-value radius. Instead, after sorting all objects by their numeric attribute, each node $v_i$ adds the $ef_\mathrm{attribute}$ predecessors and the $ef_\mathrm{attribute}$ successors around its attribute rank (lines~3 and 7--8). Thus, $ef_\mathrm{attribute}$ is an \emph{attribute-order window size}: it controls how many nearby objects in the sorted attribute order are exposed to the pruning stage. The final candidate neighbor set for each node $v$ is formed by merging its \textit{spatial candidates} from the approximate KNNG, whose degree is $ef_\mathrm{spatial}$, with these \textit{attribute-window candidates} (lines~3-8). 
Formally, for each node \(x\), let \(C_x^{sp}(V)\) be the spatial candidates generated from the dataset-level approximate KNNG and let \(C_x^{attr}(V)\) be the two-sided attribute-window candidates in the global attribute order. RNSG prunes the fixed candidate source \(C_x=C_x^{sp}(V)\cup C_x^{attr}(V)\). Relative to exact RRNG, approximation can arise from three bounded-resource effects: an ideal RRNG edge endpoint may be absent from \(C_x\), an ideal pruning witness may be absent from \(C_x\), or a valid candidate may be truncated by the finite degree budget \(m\). This is the practical gap between RRNG and RNSG.
Subsequently, we apply Algorithm~\ref{alg:RAP} to prune this candidate set, obtaining the final neighbor list for each node (line 9). Finally, the RNSG is obtained from the union of all retained edges (lines~10-12). The following theorem analyzes the RNSG-specific pruning and edge-generation cost after the spatial candidate graph is available; in the experiments, we report the KNNG and RNSG stages together as total construction time.
\reviseend


\begin{theorem} \label{thm:rnsg}
Given the spatial candidate graph, the RNSG-specific pruning and edge-generation stage in Algorithm~\ref{alg:rnsg} takes $O(|C|M_{rnsg})$ time, where $M_{rnsg}$ is the number of edges in the final RNSG and $|C|$ denotes the maximum neighborhood size. 
\end{theorem}
This follows by summing the \(O(|C|\cdot \deg_{RNSG}(v))\) pruning cost over all nodes; the full derivation is in \rnsgappendixref{app:proof-rnsg-build-time}.
In Theorem~\ref{thm:rnsg}, $|C|$ is clearly bounded by $2ef_\mathrm{attribute} + ef_\mathrm{spatial}$ (where $ef_\mathrm{spatial}$ is from NNDescent~\cite{2011nn-descent}. Since $|C| \ll n$ in practice, our method is highly efficient and scalable. 

\revRone{\noindent\textbf{Structural Properties of RNSG.} We next state the properties of the practical graph built by Algorithm~\ref{alg:rnsg}. These properties should be read together with the approximation statement above: RNSG is not the exact RRNG, but its construction still preserves the interval structure needed by range-filtered search under the fixed candidate source used to build the graph.}

\begin{theorem}[Interval Reachability]\label{thm:strongly connected}
Assume \(ef_\mathrm{attribute}\ge 1\) and the per-side retained-degree budget in Algorithm~\ref{alg:RAP} is at least one. The RNSG graph $G$ and any range-induced subgraph $G[I]$ are both strongly connected.
\end{theorem}
\begin{proof}
\revRoneStart
In Algorithm~\ref{alg:rnsg}, let the nodes be ordered by increasing attribute as \(v_1,\ldots,v_n\). Because $ef_\mathrm{attribute}\ge 1$, line~7 inserts $v_{i+1}$ into the candidate set of $v_i$ and inserts $v_i$ into the candidate set of $v_{i+1}$ for every adjacent pair. Moreover, no node has an attribute strictly between $v_i.a$ and $v_{i+1}.a$, so neither edge \((v_i,v_{i+1})\) nor edge \((v_{i+1},v_i)\) can be pruned by an attribute-between witness. Since the per-side retained-degree budget is at least one, these adjacent candidates are retained. Hence every adjacent pair in the sorted order remains connected, forming a bidirectional chain over the attribute order. Any interval-induced subgraph preserves the corresponding contiguous sub-chain, and is therefore also strongly connected.
\reviseend
\end{proof}

The strong connectivity in Theorem~\ref{thm:strongly connected} is a reachability guarantee, not a guarantee of short monotone paths or high search efficiency. It follows from the attribute-adjacent chain introduced by the attribute-window candidate source, so it does not require a large degree such as \(m=300\). Query efficiency still depends on spatial candidates, range-aware pruning, and the finite degree budget. The next result states the precise sense in which the practical RNSG construction preserves interval-restriction behavior.

\begin{theorem}[Fixed-Candidate Interval Restriction]\label{thm:rnsg-induce}
Let $G=(V,E)$ be the RNSG constructed by Algorithm~\ref{alg:rnsg} on dataset $D$. For each node $x$, let \(C_x=C_x^{sp}(V)\cup C_x^{attr}(V)\) be the fixed candidate source generated before Algorithm~\ref{alg:RAP} is applied. For any interval $I=[a_l,a_r]$, let \(V_I=\{x \in V \mid x.a \in I\}\), let \(G[I]\) be the subgraph of \(G\) induced by \(V_I\), and define \(C_x^{I,\mathrm{fix}}=C_x\cap V_I\). Let \(H[I]\) be the graph obtained by applying Algorithm~\ref{alg:RAP} to every \(x\in V_I\) using \(C_x^{I,\mathrm{fix}}\) and the same degree bound \(m\), without regenerating spatial or attribute-window candidates inside \(V_I\). Then \(G[I]=H[I]\).
\end{theorem}
\begin{proof}
\revRoneStart
Fix a node \(x\in V_I\). We show that the retained neighbor set of \(x\), when restricted to \(V_I\), is the same as the retained neighbor set produced by running Algorithm~\ref{alg:RAP} on \(C_x^{I,\mathrm{fix}}\). Consider the right side of \(x\); the left side is symmetric. If \(y\in C_x^{I,\mathrm{fix}}\) and \(y.a>x.a\), then every right-side candidate processed before \(y\) has attribute between \(x.a\) and \(y.a\), and therefore also lies in \(V_I\). Thus, before each in-range candidate is tested, the previously kept right-side set is the same in the full and restricted runs. Moreover, any valid RRNG witness for pruning \((x,y)\) must have attribute between \(x.a\) and \(y.a\), and is therefore in \(V_I\). Consequently the prune-or-keep decision for every in-range candidate is identical in the two runs. The degree cap \(m/2\) on this side is reached at the same time because all candidates that can be processed before an in-range candidate are themselves in range; once out-of-range right-side candidates appear, no later right-side candidate can return to the interval. Therefore the right-side retained set commutes with interval restriction, and so does the left-side retained set. Merging the two sides proves that the final neighbor set of every \(x\in V_I\) is identical in \(G[I]\) and \(H[I]\), and hence \(G[I]=H[I]\).
\reviseend
\end{proof}

Theorem~\ref{thm:rnsg-induce} is a fixed-candidate interval-restriction result. It does not claim that \(G[I]\) is identical to a fresh range-specific rebuild that regenerates a KNNG and attribute windows only on \(V_I\). Such a rebuild may recover additional in-range spatial candidates because some full-dataset KNNG neighbors of a node lie outside \(I\); at the same time, as intervals become smaller, attribute-window candidates form a larger relative share of the available in-range candidate source. Practical RNSG makes this trade-off once: it uses one reusable graph whose range-aware pruning decisions commute with interval restriction under the fixed candidate source, rather than rebuilding a separate graph for every query range~\cite{engels2024approximate,zuo2024serf,xu2024irangegraph}. Exp-5 quantifies this finite-candidate gap empirically.

\noindent\textbf{Extension to Multiple Ordered Attributes.}
The theory above focuses on a single totally ordered attribute because this is the standard RFANN setting and gives the cleanest interval-heredity statement. The same pruning principle \revRone{extends dimension by dimension to any fixed number} \(d\) \revRone{of ordered numeric attributes whose query predicates are conjunctions of ranges.} Let \(a_j(x)\) be the \(j\)-th numeric label of object \(x\). For an edge \((x,y)\), define \(l_j=\min\{a_j(x),a_j(y)\}\) and \(u_j=\max\{a_j(x),a_j(y)\}\). The attribute condition for a pruning witness \(z\) becomes \revRone{coordinate-wise betweenness}:
\[
  a_j(z) \in [l_j,u_j]
  \quad \text{for all } j\in\{1,\ldots,d\},
\]
in addition to the usual geometric witness condition. In the two-attribute case, this is the \revRone{rectangle-between rule}; for larger \(d\), it becomes the corresponding \revRone{axis-aligned hyperrectangle-between rule}. Candidate generation likewise adds attribute-window candidates along each sorted numeric attribute order, in addition to the spatial KNNG candidates, and query processing applies all range predicates during traversal. The structural reason is that if two endpoints satisfy a query box and a witness is coordinate-wise between them, then the witness also satisfies that query box. Thus the same witness-preservation intuition behind one-dimensional interval filtering carries over to conjunctions of numeric range filters. \revRone{The main theory remains one-dimensional;} \revRone{Exp-7 validates the two-attribute extension.}


\subsection{Query Processing Algorithm With RNSG} \label{subsec:query-pro-rnsg}
RFANN query processing on RNSG uses beam search, following standard graph-based ANN indexes~\cite{wang2021comprehensive,25SIGMODgraphindexsurvey}. Crucially, for any range query $I$, the algorithm does not explicitly materialize the induced subgraph. Instead, the search is performed on the entire RNSG, with the range filter applied on-the-fly during traversal to prune edges that lead to out-of-range nodes. This differs from conventional \emph{in-filter} methods. Traditional in-filter approaches often suffer from degraded navigability and even graph disconnection, as range filters prune edges outside the query interval, drastically reducing the effective connectivity of the search graph. In contrast, RNSG is supported by two structural facts established above: every interval-induced subgraph remains strongly connected (Theorem~\ref{thm:strongly connected}), and range-aware pruning decisions among in-range candidates are stable under the fixed candidate source (Theorem~\ref{thm:rnsg-induce}). Together, these results explain why on-the-fly range filtering in RNSG is much less disruptive than filtering generic ANN graphs.

\noindent\textbf{Entry Node Generation.} As observed in previous studies \cite{25SIGMODgraphindexsurvey, xu2024irangegraph,wang2017survey}, the starting nodes (entry nodes) of the beam search affect search performance. Traditional graph-based ANN methods typically use the \textit{center node}, i.e., the node closest to the centroid of the entire dataset~\cite{wang2017survey,fu2017fast}. For RFANN queries, however, \revRone{the best centroid-like entry changes with the range constraint.} Precomputing such an entry for every possible range would require impractical $O(n^2)$ space.

To avoid this quadratic cost, we use a compact \revRone{\emph{header} entry policy.} The structure is precomputed once by scanning objects in attribute order and keeping suffix minima with respect to the distance to the global centroid. At query time, for a range $[a_l,a_r]$, we look up the stored list associated with the right endpoint and choose the first object that also satisfies the left endpoint. This adds only a lightweight lookup before beam search. \revRone{The implementation stores the header as compact arrays embedded in the deployed graph file, not as a separate range-specific graph or cover index.} The detailed procedure and example are given in \rnsgappendixref{app:entry-demo}. \revRone{The independence assumption} in Lemma~\ref{lem:entrysize} is used only to explain the expected size of this entry structure; \revRone{it is not required for correctness,} because beam search can start from any in-range node.

\begin{lemma}\label{lem:entrysize}
Assume that the vector embedding distribution is independent of the attribute distribution. Then the expected number of distinct entry nodes required for all query ranges with a fixed right endpoint $a_r$ is $O(\log n)$.
\end{lemma}

\revRone{The structure is a query-time accelerator, not a correctness prerequisite;} the ablation studies its practical effect.


\begin{table}[t]
\small
    \centering
    \vspace*{-3em}    
    \caption{\revRone{Dataset and Workload Statistics}}
    \label{tab:dataset}
    \vspace*{-0.3em}
    \renewcommand{\arraystretch}{1.2}
    \setlength{\arrayrulewidth}{0.8pt}
    \begin{tabularx}{0.45\textwidth}{
        |>{\centering\arraybackslash}p{1.7cm}
        |>{\centering\arraybackslash}X
        |>{\centering\arraybackslash}X   
        |>{\centering\arraybackslash}X        
        |>{\centering\arraybackslash}p{1.5cm}|
        }
        \hline
        \textbf{Dataset} & $|D|$& \#Query & $d$ & Vector Type \\
        \hline 
        YT-Audio& 1M & 1,000 &128  &audio \\        
        YT-RGB&  1M & 1,000 &1,024  &video \\        
        WIT&  1M & 1,000 &2,048  & image \\
        {TripClick}& {1M} & {1,000} &{768} &{text} \\
        SIFT1M & 1M & 10,000 & 128 & image   \\
        GIST1M& 1M & 1,000 &960  &image   \\
        {SpaceV10M}& 10M & 1,000 &100 &image \\
        {SIFT-scale}& 10--50M & 10,000 &128 &image \\
        \hline        
    \end{tabularx}
\end{table}
\section{Experiments} \label{sec:exp}

\begin{figure*}[!t]
\centering
\vspace*{-3em}
\includegraphics[width=0.9\textwidth]{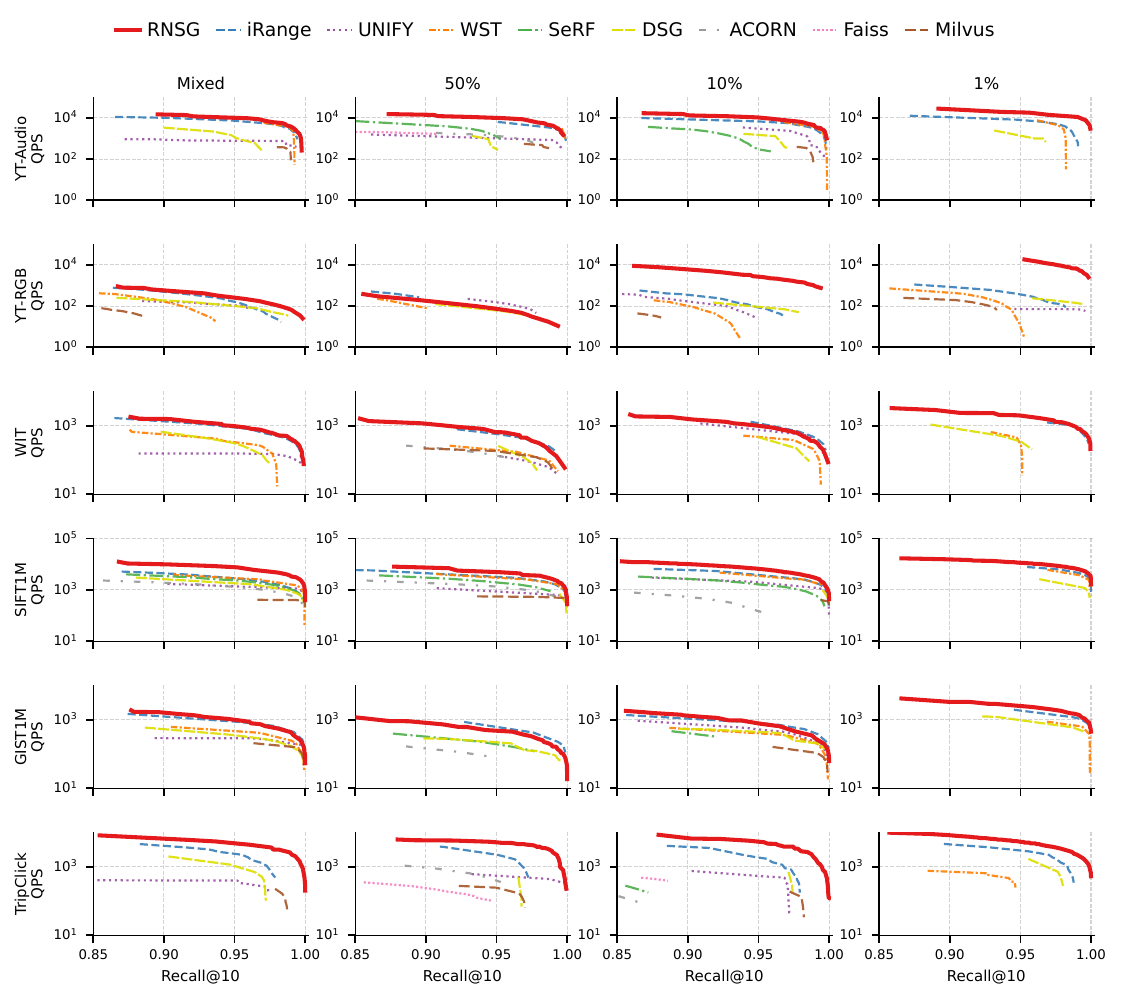}
\vspace*{-1em}
\caption{Recall-QPS curves on six datasets under mixed, 1\%, 10\%, and 50\% selectivity workloads (upper and right is better).}
\vspace*{-1em}
\label{fig:exp-acc-eff}
\end{figure*}
\subsection{Experimental Setup} \label{subsec:exp-setup}
\noindent\textbf{Datasets.} Following previous RFANN studies~\cite{xu2024irangegraph,liang2024unify,zuo2024serf,li2025attribute}, we use \revRone{6} widely-used real datasets as the 
benchmark: YouTube-Audio 
(YT-Audio for short), YouTube-RGB\footnote{\url{https://research.google.com/youtube8m/download.html}} (YT-RGB for short), 
WIT-Image\footnote{\url{https://github.com/google-research-datasets/wit}} (WIT for short),  
\revRone{TripClick}\footnote{\url{https://tripdatabase.github.io/tripclick/}}, SIFT1M and GIST1M\footnote{\url{http://corpus-texmex.irisa.fr/}}. \revRone{Table~\ref{tab:dataset} also includes SpaceV10M and {SIFT-scale}, two larger datasets used for scalability evaluation.} 
For numeric attribute assignment, we maintain consistency with previous studies~\cite{xu2024irangegraph,li2025attribute}: YT-Audio uses video publication time, YT-RGB employs \textit{like} counts, WIT utilizes image size, \revRone{TripClick uses timestamp}, and SIFT1M and GIST1M use synthetic numeric attributes generated by random distributions following the methodology used in~\cite{li2025attribute}. 
\noindent\textbf{Query Ranges.} Consistent with prior research~\cite{xu2024irangegraph}, we note that despite the widespread industrial adoption of RFANN queries (e.g., at Apple~\cite{pound2025micronn}, Milvus~\cite{wang2021milvus}, and Alibaba~\cite{wei2020analyticdb}), no benchmark datasets and query ranges are publicly available due to privacy considerations. Following established evaluation methodologies~\cite{xu2024irangegraph,li2025attribute}, we generate synthetic query ranges to assess performance. For a fixed selectivity \(s\), each query range is a contiguous interval in the sorted attribute order containing approximately an \(s\) fraction of the dataset, and all methods use the same range files and exhaustive ground-truth results. We evaluate queries through two query workloads: (1) a mixed workload for overall performance evaluation, using the fixed mixed-selectivity query set; and (2) fixed selectivity workloads for specific performance analysis. The fixed-selectivity curves report 1\%, 10\%, and \revRone{50\%} selectivity to evaluate method performance across low-, medium-, and \revRone{high-selectivity} RFANN workloads.

\noindent\textbf{Evaluation Metrics.} Following previous studies~\cite{wang2021comprehensive,wang2017survey}, we use two standard metrics to assess method performance. Accuracy is measured by $\text{recall@}k = \frac{|\mathcal{R}\cap \tilde{\mathcal{R}}|}{k}$, where $\mathcal{R}$ represents the result set retrieved by the evaluated method and $\tilde{\mathcal{R}}$ denotes the ground-truth set obtained through exhaustive linear scan. Consistent with recent RFANN studies~\cite{liang2024unify,xu2024irangegraph,li2025attribute}, we primarily report recall@10 to facilitate direct comparison with state-of-the-art methods, \revRone{and Exp-7 further compares methods while varying} $k$. For efficiency measurement, we use queries per second (QPS), calculated as $\text{QPS} = \frac{|Q|}{t}$, where $|Q|$ queries are processed within time $t$.  

\noindent\textbf{Compared Methods.} We evaluate our RNSG method against 8 representative baselines for RFANN search: (1) iRangeGraph~\cite{xu2024irangegraph}: Integrates segment trees with graph indexes, constructing elemental graphs for multiple ranges and dynamically composing them during query processing (Section~\ref{sec:preliminary}). (2) UNIFY~\cite{liang2024unify}: Partitions data by numeric attributes, builds intra-segment graphs, and employs skip lists with edge bitmaps for efficient range pruning (Section~\ref{sec:preliminary}). (3) SuperPostfiltering~\cite{engels2024approximate}, denoted WST in the plots: a post-filtering approach that pre-builds graph indexes for multiple overlapping ranges and selects the minimal covering range at query time (Section~\ref{sec:preliminary}). (4) SeRF~\cite{zuo2024serf}: Employs aggressive graph compression to reduce indexing costs while maintaining search performance. We use the MaxLeap variant reported in the original paper. (5) DSG~\cite{peng2025dynamic}: A recent graph method for dynamic RFANN workloads that uses rectangle trees for range partitioning and builds upon SeRF-style compression. We include DSG because it is conceptually close and construction-efficient, while noting that it is primarily designed for dynamic rather than static workloads. (6) ACORN-$\gamma$~\cite{patel2024acorn}: A predicate-agnostic hybrid ANN search method applicable to arbitrary filters, though less optimized for RFANN queries specifically. (7) Faiss-IVFPQ~\cite{douze2024faiss}: A widely-used vector database baseline that filters out-of-range data points during retrieval using inverted file indexing with product quantization. (8) Milvus-HNSW~\cite{wang2021milvus}: Employs HNSW indexing with automatic strategy selection among pre-filtering, in-filtering, and post-filtering approaches based on cost estimation (Section~\ref{sec:preliminary}). Among the evaluated methods, iRangeGraph~\cite{xu2024irangegraph} and UNIFY~\cite{liang2024unify} represent the current state-of-the-art in RFANN search, as established by a recent benchmarking study~\cite{li2025attribute}. We also note two recent methods, RangePQ~\cite{zhang2025efficient} and DIGRA~\cite{jiang2025digra}, which are designed for dynamic RFANN scenarios with efficient index updates. However, their experiments target update-heavy settings rather than the static RFANN setting studied here, so we exclude them from our comparisons.

\noindent\textbf{Parameter Settings.} The construction of our RNSG index involves three key parameters (Section~\ref{subsec:rnsg-construction-alg}): $ef_\mathrm{spatial}$ controls the size of the spatial-similar neighborhood, $ef_\mathrm{attribute}$ bounds the size of the attribute-window neighborhood in sorted attribute order, and $m$ specifies the maximum final neighborhood size. \revRone{Unless otherwise stated, RNSG uses $ef_\mathrm{spatial}=128$, $ef_\mathrm{attribute}=1500$, and $m=200$; GIST1M uses $m=300$ because the high-dimensional dataset benefits from the larger final neighborhood. Exp-4 studies the sensitivity of 
parameters and gives practical guidance for setting them. For the baseline methods, we adopt the hyperparameter tuning strategy used in FANNBench~\cite{li2025attribute} to obtain dataset-specific settings. Due to page limitations, we defer the detailed hyperparameter configurations to the technical report~\cite{zou2026rnsgrangeawaregraphindex}}.

\begin{figure*}[!t]
\centering
\includegraphics[width=0.9\textwidth]{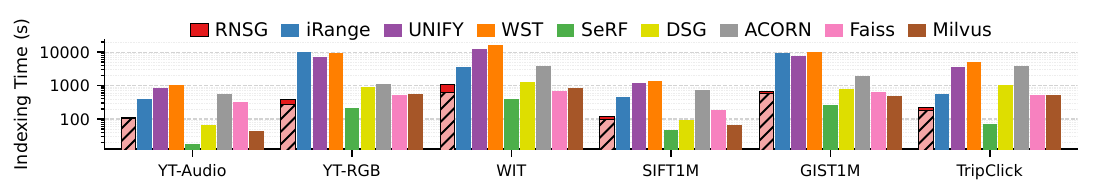}
\vspace*{-1em}
\includegraphics[width=0.9\textwidth]{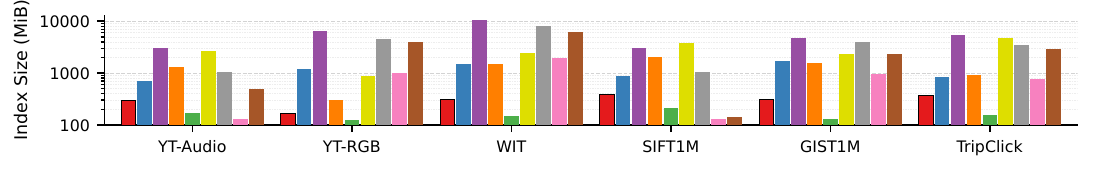}
\vspace*{-0.2em}
\caption{\revRone{Construction cost and deployed index size} on the original benchmark datasets. Top: index construction time, \revRone{where RNSG is shown as KNNG construction plus the RNSG-specific build stage.} Bottom: deployed index size.}
\vspace*{-1em}
\label{fig:index-cost}
\label{fig:index-time}
\label{fig:memory}
\end{figure*}



\noindent\textbf{Experimental Environment.} All methods are implemented in C++ and compiled with GCC 15.1.0 using -O3. We use the public baseline implementations, measure indexing with 128 threads, and report single-thread query performance. Experiments run on a server with an AMD Ryzen Threadripper 3990X Processor (2.2GHz) and 224 GB of memory. 

\vspace{-1em}
\subsection{Experimental Results} \label{subsec:exp-results}
\noindent\textbf{Exp-1: RFANN Querying Performance.} We compare RNSG with the baselines in Fig.~\ref{fig:exp-acc-eff}. The benchmark includes \revRone{TripClick}; \revRone{the fixed-selectivity workloads include 50\% high-selectivity queries.} The main findings are summarized as follows.

\textit{\textbf{(1) Strong Overall Performance}}: RNSG achieves the best or near-best performance under both mixed and fixed query workloads on most datasets. \revRone{The recall-QPS curves include DSG with dataset-specific construction profiles.} Across the evaluated static RFANN workloads, RNSG provides strong matched-recall trade-offs while keeping the deployed index as one graph.

\revRone{
\textit{\textbf{(2) Strength Across Selectivities}}: Our method performs particularly well on low-selectivity queries because RNSG preserves desirable search properties within the induced subgraph for any query range, while other methods face trade-offs between pruning excessive edges and evaluating out-of-range objects.}
\revRone{At 50\% selectivity, RNSG remains competitive on broader high-selectivity RFANN workloads.} As selectivity increases, fewer intermediate vertices are removed by the range predicate, so the workload becomes closer to ordinary ANN search and the range-induced navigability problem becomes less severe for all graph-based methods. \revRone{RNSG still uses part of its degree budget to preserve range-aware connectivity across arbitrary intervals; under broad ranges, these range-supporting edges provide less marginal benefit and can increase the number of in-range neighbors examined during traversal. Thus, the speedup over strong baselines naturally narrows at high selectivity, while RNSG remains competitive and retains the single-index storage advantage.}

\textit{\textbf{(3) Broader baseline behavior}}: The broader benchmark in Fig.~\ref{fig:exp-acc-eff} also clarifies the role of compact and generic filtered-ANN alternatives. SeRF is compact but can lose recall due to aggressive compression, while predicate-agnostic systems such as ACORN-$\gamma$, Faiss-IVFPQ, and Milvus-HNSW are less specialized for ordered numeric ranges. These trends are consistent with previous RFANN studies~\cite{xu2024irangegraph, li2025attribute} and show where RNSG's range-aware pruning provides the most effective recall-QPS trade-off.

\begin{figure}[!t]
  \centering
  \includegraphics[width=\linewidth]{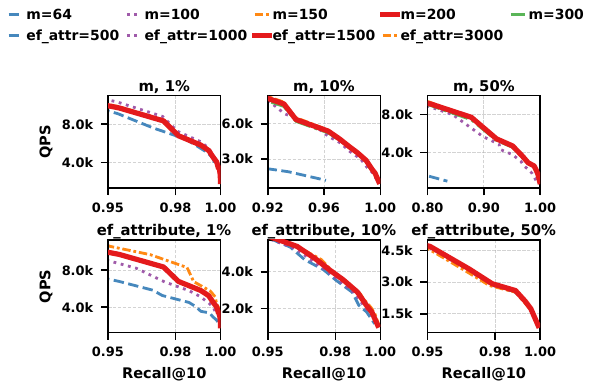}
  \caption{\revRone{Parameter sensitivity on SIFT1M.} The top row varies the final-neighborhood budget $m$, the bottom row varies the attribute-window budget, and columns correspond to 1\%, 10\%, and 50\% selectivity.}
  \vspace*{-1em}
  \label{fig:para-affect}
\end{figure}

\noindent\textbf{Exp-2: Index Construction Time}. The top panel of Fig.~\ref{fig:index-cost} compares construction time on the original benchmark datasets. For RNSG, \revRone{we report the full one-time cost as KNNG construction plus the RNSG-specific build stage.} \revRone{The KNNG is a reusable spatial candidate graph for a fixed dataset, metric, and degree: it can be built once for multiple RNSG parameter settings, shared with other graph-index pipelines, and discarded after the final RNSG graph is produced.} The figure therefore makes build cost explicit. DSG is construction-efficient, consistent with its dynamic-oriented design; RNSG's main advantage is the combined query-performance and single-graph storage trade-off.

\noindent\textbf{Exp-3: Index Size.} The bottom panel of Fig.~\ref{fig:index-cost} reports deployed index size. \revRone{RNSG stores one sparse graph, whereas several RFANN baselines obtain range awareness by materializing or composing segment-level graphs, range-cover indexes, auxiliary edge sets, or range-pruning metadata.} Thus, a smaller base-graph degree does not necessarily imply a smaller deployed index when connectivity is replicated across many partitions or range covers. \revRone{RNSG instead embeds attribute-order information into the pruning rule, so range-supporting edges are retained in the same adjacency lists and shared by all query intervals; the temporary KNNG is not deployed.} Directly parsing the deployed 1M graph files confirms that the degree budget is only an upper bound: their realized average out-degree is 41.6--100.3 under the reported budgets, and \revRone{the compact header entry arrays average 9.8--10.5 ids per header, occupying only 0.65--5.07 MiB, or 0.21\%--1.65\% of the serialized graph file.} SeRF has the smallest footprint due to aggressive compression, but its space result should be read together with its recall-QPS curves. Overall, RNSG targets a favorable memory-performance trade-off rather than minimizing index size alone.

\noindent\textbf{\revRone{Exp-4: Parameter Sensitivity Analysis.}}
This experiment investigates the construction-side parameters that determine RNSG's graph quality and index cost. Fig.~\ref{fig:para-affect} separates the final-neighborhood budget $m$ from the attribute-window budget instead of mixing all selectivities into one aggregate point. This is important because a mixed workload can hide the behavior of small-range queries. \revRone{The same sweep also provides practical evidence for the finite-candidate approximation gap: if increasing the candidate or degree budget no longer changes recall-QPS or realized degree meaningfully, the bounded construction has reached the useful part of the available range-aware candidate source for this workload. We use this only as a saturation check, not as a claim of exact RRNG edge equivalence.}

The parameter $m$ is the maximum number of neighbors retained for each node after range-aware pruning. Increasing $m$ gives the search more preserved directions and can improve high-recall connectivity under range filters, but it also increases graph degree, index size, construction time, and the number of neighbors considered during search. Its effect is therefore expected to saturate: $m$ is only an upper bound on the retained degree, and the pruning rule discards candidates that are already covered by closer range-aware witnesses. \revRone{On SIFT1M, $m=64$ is too small: it misses the 0.98 recall target at both 10\% and 50\% selectivity, with maximum recall 0.9613 and 0.8351 respectively. Increasing $m$ to 100 makes the graph usable, while $m=150$, 200, and 300 produce close recall-QPS curves. We therefore use $m=200$ as the default, since the benefit has largely saturated without paying the extra cost of the larger graph.}

The varying-$k$ experiment below uses the same $m=200$ graph for $k=20$ to $100$. \revRone{This saturation is also visible in the stored graph: increasing the degree budget from 200 to 300 changes the SIFT1M average out-degree only from 100.34 to 100.59 and the serialized graph size from 391.7 to 392.6 MiB.} For the attribute-window sweep, $ef_\mathrm{attribute}=1500$ is the stable default. A larger window, $ef_\mathrm{attribute}=3000$, can improve the 1\% boundary, but it increases graph degree, index size, and build cost while bringing little benefit at 10\% and 50\%. \revRone{For example, on SIFT1M it raises the realized average out-degree from 100.34 to 110.18 and the graph size from 391.7 to 429.2 MiB, so the larger candidate window is not the default unless the workload specifically requires the 1\% boundary gain.}

\noindent\textbf{\revRone{Exp-5: Approximation-Gap Probe.}}
\revRone{This diagnostic experiment quantifies the cost of using one reusable graph rather than rebuilding a graph for each query interval. 
For 10 sampled 1\% intervals on TripClick1M and 50 queries per interval, we compare the deployed Full-RNSG with a diagnostic Range-Rebuild reference that regenerates the KNNG and attribute-window candidates inside each interval using the same construction budget ($ef_\mathrm{spatial}=128$, $ef_\mathrm{attribute}=1500$, and $m=200$). 
Both methods use the same query vectors, intervals, deterministic in-range entry policy, and exhaustive in-range ground truth; rebuild construction time is excluded because this is a diagnostic reference, not a deployable baseline. 
Structurally, the interval-induced subgraphs of Full-RNSG cover 81.5670\% of the directed edges in the corresponding Range-Rebuild graphs on average. 
This coverage is stable across the 10 intervals, ranging from 81.3792\% to 81.7648\%, with no apparent outlier interval. 
\revRone{After a dense local beam/truncation sweep around the recall@10 target 0.98, the closest observed matched-recall pair gives recall 0.9834 for both methods: Full-RNSG uses 1{,}004 distance computations on average, while Range-Rebuild uses 651.} 
Thus, the single reusable graph preserves most interval-specific directed edges but still pays about $1.54\times$ the distance evaluations of a per-interval rebuild at the same recall level, which quantifies the finite-candidate approximation gap while explaining the practical trade-off: RNSG avoids materializing a separate graph for every possible range.}

\begin{figure}[!t]
  \vspace*{-3em}
  \centering
  \includegraphics[width=\linewidth]{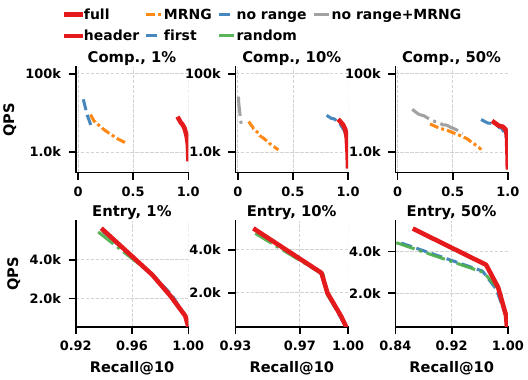}
\caption{\revRone{Design ablation on TripClick.} The top row removes construction components, the bottom row compares entry policies on the full construction, and columns correspond to 1\%, 10\%, and 50\% selectivity.}
  \label{fig:design-ablation}
\end{figure}

\noindent\textbf{\revRone{Exp-6: Design Ablation.}} Fig.~\ref{fig:design-ablation} evaluates RNSG's main design components: the construction ablations test \revRone{the attribute-window augmentation} in Algorithm~\ref{alg:rnsg} and \revRone{RAP's range-aware pruning}, while the entry ablations isolate the query-time startup policy. 
On TripClick, removing attribute-window augmentation severely hurts low-selectivity search, \revRone{reaching only 0.1148 maximum recall on the 1\% workload} and also missing the 0.98 recall target at 10\%. 
\revRone{Replacing RAP with conventional MRNG pruning} fails to reach 0.98 recall at all tested selectivities, with \revRone{maximum recalls of 0.4350, 0.3678, and 0.7580 at 1\%, 10\%, and 50\%, respectively}; disabling both components degrades the graph further. 
For entry policies, header, first, and random entries give generally close curves, but header can still provide up to a 1.1\(\times\) QPS gain with almost no overhead. 
We therefore present header entry as a stable, lightweight startup heuristic for cheap additional improvement, rather than as the primary source of RNSG's speedup.

\begin{figure}[!t]
  \vspace*{-3em}
  \centering
  \includegraphics[width=\linewidth]{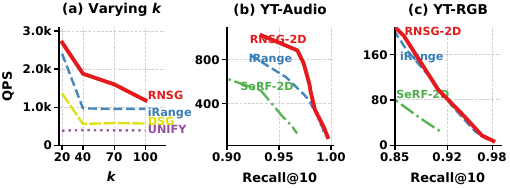}
  \caption{\revRone{Additional robustness and extension experiments:} (a) varying-$k$ on TripClick; (b,c) two-attribute range-filtered search on YouTube datasets.}
  \label{fig:followup-compact}
\end{figure}


\noindent\textbf{\revRone{Exp-7: Robustness to} $k$ \revRone{and Multiple Numeric Attributes.}}
\revRone{ Fig.~\ref{fig:followup-compact}(a) evaluates varying-\(k\) RFANN queries on TripClick at 10\% selectivity and recall at least 0.95, comparing methods at matched recall along their recall-QPS curves. 
RNSG remains faster than iRangeGraph by \(1.12\times\), \(1.94\times\), \(1.68\times\), and \(1.24\times\) for \(k=20\), 40, 70, and 100\revRone{, respectively,} indicating that the query-performance conclusion is not limited to recall@10. 
Fig.~\ref{fig:followup-compact}(b,c) further evaluates RNSG-2D on YouTube two-attribute workloads with conjunctive ranges over \textit{likes} and \textit{views}, using the FANNBench multi-query interface and the same query ranges and exhaustive ground truth as the iRangeGraph protocol. 
iRangeGraph is evaluated through its native FANNBench implementation, while SeRF-2D follows the SeRF-style two-attribute comparison used in the iRangeGraph evaluation.
RNSG-2D outperforms the baselines on YT-Audio and achieves comparable high-recall trade-offs to iRangeGraph on YT-RGB, whereas SeRF-2D does not reach the same high-recall regime. 
This experiment provides a concrete \revRone{two-attribute validation} of the \revRone{coordinate-wise pruning principle}, rather than a complete many-attribute system study.}

\begin{figure}[!t]
  \centering
  \begin{subfigure}[t]{0.30\linewidth}
    \centering
    \includegraphics[page=1,width=\linewidth]{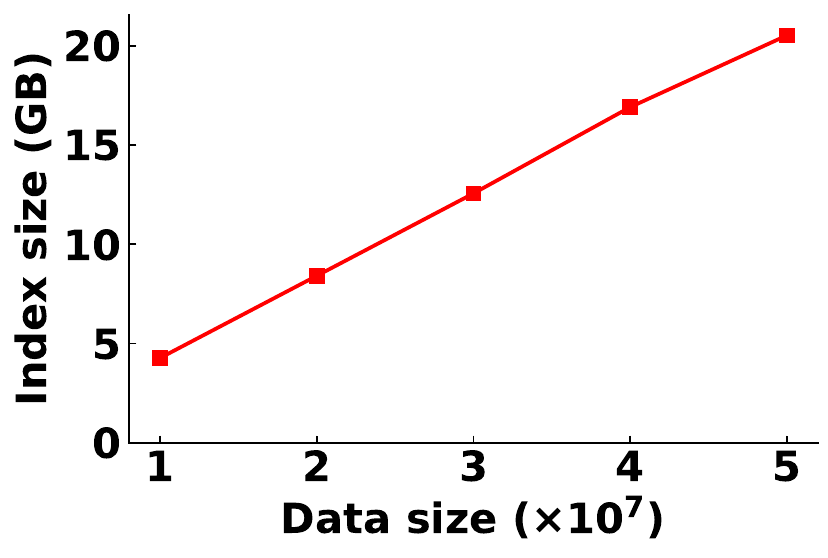}
    \caption{\revRone{Index size}}
  \end{subfigure}%
  \hfill
  \begin{subfigure}[t]{0.30\linewidth}
    \centering
    \includegraphics[page=1,width=\linewidth]{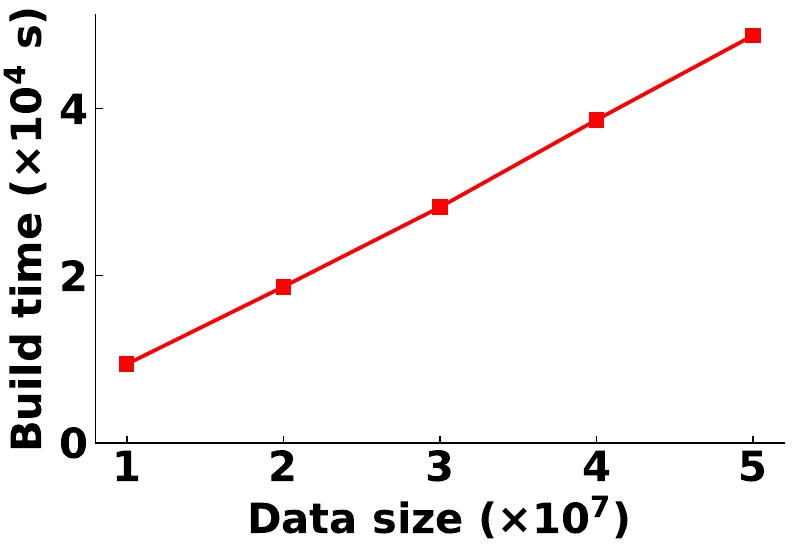}
    \caption{\revRone{Indexing time}}
  \end{subfigure}%
  \hfill
  \begin{subfigure}[t]{0.32\linewidth}
    \centering
    \includegraphics[page=1,width=\linewidth]{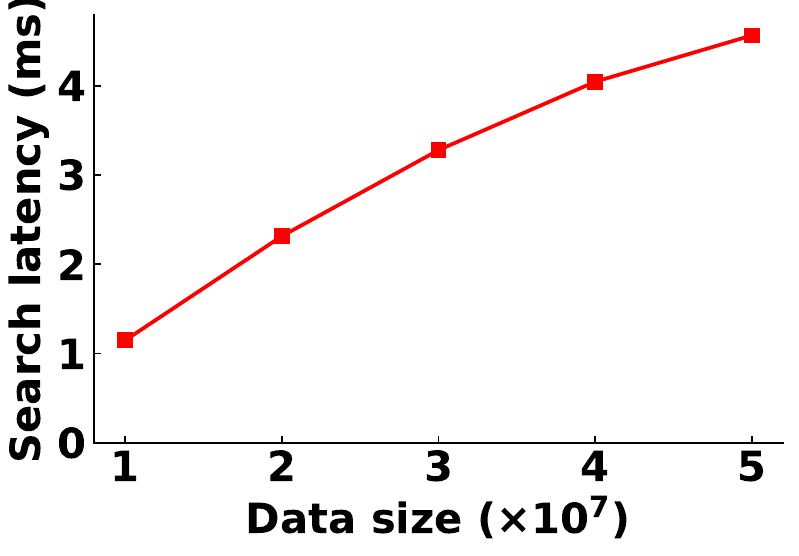}
    \caption{\revRone{Latency}}
  \end{subfigure}
  \vspace*{-0.5em}
  \caption{\revRone{Cardinality scalability} of RNSG on SIFT10M--SIFT50M. Query latency is reported at 0.9 recall.}
  \vspace*{-1em}
  \label{fig:scalability}
\end{figure}

\noindent\textbf{\revRone{Exp-8: Scalability on Larger Datasets.}} We further isolate \revRone{cardinality scalability} on SIFT-scale workloads from 10M to 50M objects using the mixed query workload. Fig.~\ref{fig:scalability} shows approximately linear growth in index size and construction time, with smooth latency growth. This complements SpaceV10M by isolating RNSG's own scaling as cardinality grows, which can be found in appendix.


\section{Related Work} \label{sec:related-work}
\noindent\textbf{Approximate Nearest Neighbor Search.} Existing ANN algorithms fall into four main categories: graph-based~\cite{25SIGMODgraphindexsurvey,wang2021comprehensive,malkov2018efficient,malkov2014approximate,harwood2016fanng,fu2017fast,azizi2023elpis,jayaram2019diskann,chen2021spann}, quantization-based~\cite{DBLP:conf/mm/TuncelFR02,gao2025practical,ge2013optimized,jegou2010product,gao2024rabitq,gong2012iterative,matsui2018survey,aguerrebere2023similarity,andre2016cache,li2025saq}, hashing-based~\cite{wang2014hashing,wang2017survey,pham2022falconn,DBLP:journals/tods/TaoYSK10,LSHDatarIIM04,DBLP:journals/pami/HeoLHCY15,DBLP:journals/pvldb/HuangFZFN15,DBLP:conf/sigmod/GanFFN12,DBLP:conf/sigmod/LeiHKT20,DBLP:conf/sigmod/LiYZXCLNC18}, and tree-based methods~\cite{ram2019revisiting,beygelzimer2006cover}. For a comprehensive overview, we refer readers to recent surveys~\cite{pan2024vector,DBLP:journals/vldb/PanWL24}. Extensive empirical evaluations~\cite{liApproximateNearestNeighbor2020,DBLP:journals/is/AumullerBF20,wang2017survey,25SIGMODgraphindexsurvey} have consistently shown that graph-based methods achieve state-of-the-art performance. This superiority is attributed to the graph's ability to encode proximity relationships, enabling query processing to converge rapidly by examining only a small fraction of the dataset~\cite{dobson2023scaling}. Most graph-based indexes~\cite{wang2021comprehensive} are built upon four fundamental graph types: the Delaunay Graph (DG)~\cite{dgraph}, \yin{Monotonic} Relative Neighborhood Graph (\yin{MRNG})~\cite{fu2017fast}, K-Nearest Neighbor Graph (KNNG)~\cite{knngraph}, and Minimum Spanning Tree (MST)~\cite{mst}. Among these, \yin{MRNG}-based indexes deliver particularly high performance due to their effective pruning strategy~\cite{wang2021comprehensive}. Building upon the \yin{MRNG}, this work introduces a novel range-aware 
\yin{ANN} index, specifically designed to enable efficient RFANN search.

\noindent\textbf{Attribute-filtered ANN Search.} Most existing RFANN studies~\cite{zhang2025efficient,wang2021milvus,douze2024faiss} follow a two-stage paradigm that separates vector and attribute retrieval before merging results. For example, ADBV~\cite{wei2020analyticdb} uses a B-Tree and a PQ index separately, choosing a strategy via a cost model. A fundamental weakness of these methods is that they do not co-optimize the index structure, resulting in performance that is often inferior to modern hybrid indexes \cite{chronis2025filtered,xu2024irangegraph,zuo2024serf,liang2024unify,engels2024approximate,li2025attribute}. Recent hybrid indexes like iRangeGraph \cite{xu2024irangegraph} integrate graph and attribute indexes more directly by building a graph for each segment of a partitioned attribute range. While this yields high performance, it does so at the cost of significant indexing overhead. Beyond RFANN-specific methods, general-purpose systems like Faiss~\cite{douze2024faiss}, Milvus~\cite{wang2021milvus}, VBASE~\cite{zhang2023vbase}, and ACORN~\cite{patel2024acorn} handle filtered queries on unordered attributes. Since they do not leverage the ordering of numeric attributes in their index structures or algorithms, they tend to be outperformed by RFANN-optimized methods in prior studies and in our evaluation~\cite{xu2024irangegraph,li2025attribute}.

\vspace*{-0.5em}
\section{Conclusion and Future Work} \label{sec:conclusion}
In this paper, we propose a novel graph index for efficient RFANN 
search. We first establish a foundational graph indexing theory, the Range-Aware Relative Neighborhood Graph (RRNG), which integrates both spatial and attribute proximity of the data objects. The RRNG guarantees that beam search can find the exact nearest neighbor for any node, under any query range. As exact RRNG construction is computationally expensive, we propose RNSG, a practical solution that approximates the RRNG efficiently. Comprehensive experiments 
demonstrate that RNSG outperforms state-of-the-art methods, achieving superior index compactness, faster construction times, and significantly higher query performance. Future work includes extending our graph index to support dynamic updates and non-numerical attribute filters.

\newpage
\balance
\bibliographystyle{ACM-Reference-Format}
\bibliography{ref}

\ifrnsgfullversion
\newpage
\appendix

\comment{
\section{Proof of Lemma~\ref{lem:efficient_search}}
\begin{proof}
According to Lemma~\ref{thm:rrng-monotone}, we can conclude that for any starting point $x$ and endpoint $y$, $x$ can always find a neighbor $z$ such that $z$ is closer to $y$ than $x$ is to $y$, and there must also exist a corresponding node from $z$ to $y$. 

Consider the beam-search process: the candidate queue always contains the closest set of points to $y$ among all previously traversed points. Thus, during the search, the leading point in the queue either directly finds $y$ itself or, within one expansion step, discovers a point closer to $y$ and places it at the queue head. Consequently, as beam-search continues, it inevitably finds the node $y$.
\end{proof}
}

\section{Appendix} \label{sec:appendxi}
\subsection{Baseline Configuration Details}\label{app:baseline-config}
Table~\ref{tab:baseline-config} summarizes the construction-parameter policy used for the baseline comparison in Section~\ref{sec:exp}. All baselines are invoked through the FANNBench wrappers or their original interfaces as integrated in FANNBench. Query-time search parameters are swept to produce recall-QPS curves; the table below reports construction-side settings because those determine build time and index size.

\begin{table*}[t]
\small
\centering
\caption{Construction-parameter policy for the evaluated baselines. Dataset-specific exceptions follow the FANNBench profiles and the cleaned construction metadata used for the paper figures.}
\label{tab:baseline-config}
\begin{tabularx}{0.98\textwidth}{lX}
\toprule
Method & Construction policy \\
\midrule
RNSG & $ef_\mathrm{spatial}=128$, KNNG degree $128$, $ef_\mathrm{attribute}=1500$, and incoming reverse refinement. We use $m=200$ by default and $m=300$ for GIST1M according to the parameter study. Reported build time is KNNG time plus the RNSG-specific pruning/build stage; the deployed index stores only the final RNSG graph and query-time auxiliary data. \\
iRangeGraph & FANNBench iRangeGraph profile. YT-RGB and GIST1M use $M=64$ and $ef_\mathrm{construction}=400$; the other datasets use $M=16$ and $ef_\mathrm{construction}=100$. \\
DSG & Official DSG/FANNBench wrapper with compact construction. We run a dataset-specific construction grid and rerun the final query curves from the selected profiles: YT-Audio/SIFT1M use $M=16$, GIST1M uses $M=32$, YT-RGB/WIT/TripClick use $M=64$, all use $ef_\mathrm{construction}=160$, and $ef_\mathrm{max}=300$ except TripClick, which uses $ef_\mathrm{max}=600$. \\
UNIFY & FANNBench UNIFY profile with $ef_\mathrm{construction}=500$ and $B=8$. GIST1M uses $M=16$; the other datasets use $M=40$. \\
SuperPostfiltering (WST) & FANNBench WST profile with split factor $2$, shift factor $0.5$, and final beam multiplier $32$. \\
SeRF & FANNBench/MaxLeap profile. Most datasets use $M=8$, $ef_\mathrm{max}=1000$, and $ef_\mathrm{construction}=500$; SIFT1M uses $ef_\mathrm{construction}=300$; GIST1M uses $M=32$, $ef_\mathrm{max}=100$, and $ef_\mathrm{construction}=100$. \\
ACORN-$\gamma$ & FANNBench ACORN profile. YT-Audio, YT-RGB, SIFT1M, and TripClick use $M=40$, $M_\beta=64$, and $\gamma=25$; WIT and GIST1M use $M=32$, $M_\beta=32$, and $\gamma=10$. \\
Faiss-IVFPQ & FANNBench Faiss-IVFPQ profile with dataset-dependent partition size: $128$ for YT-Audio/SIFT1M, $768$ for TripClick, $960$ for GIST1M, $1024$ for YT-RGB, and $2048$ for WIT. \\
Milvus-HNSW & FANNBench Milvus-HNSW profile. GIST1M uses $M=32$ and $efConstruction=400$; the other datasets use $M=40$ with dataset-dependent $efConstruction$ values from the FANNBench profile. \\
\bottomrule
\end{tabularx}
\end{table*}

\subsection{Proof of Corollary~\ref{lem:efficient_search}}\label{app:proof-efficient-search}
\begin{proof}
Theorem~\ref{thm:rrng-monotone} guarantees the existence of a directed monotonic path from any entry node $x$ to any target object $y$: at each step, there is an outgoing neighbor closer to $y$. In best-first graph search, the candidate min-heap $\mathcal{S}$ maintains visited nodes prioritized by their distance to $y.v$. In each iteration, the closest visited node is expanded. If it is not yet $y$, monotonicity guarantees that at least one outgoing neighbor is closer to $y$ and will be inserted into $\mathcal{S}$. With a search list large enough not to discard the current closest candidate, the minimum distance in $\mathcal{S}$ strictly decreases until $y$ is reached.
\end{proof}

\subsection{Proof of Theorem~\ref{thm:rrng-built-time}}\label{app:proof-rrng-build-time}
\begin{proof}
The direct exact construction considers every source node $x$, every candidate $y$, and every earlier same-side candidate $z$ that could be an attribute-between witness for the directed edge $x\to y$. This enumerates $O(n^3)$ source-candidate-witness triples. The geometric comparisons cost $O(d)$ each, but in the intended high-dimensional ANN setting this factor is dominated by the cubic witness enumeration in the graph-construction analysis, giving the stated $O(n^3)$ bound.
\end{proof}

\subsection{Proof of Lemma~\ref{lemma:prune-divide}}\label{app:proof-prune-divide}
\begin{proof}
For any $v_l\in C_l$ and $v_r\in C_r$, we have $v_l.a<x.a<v_r.a$. Hence neither side can satisfy the attribute-between condition required to witness an edge from $x$ to a candidate on the opposite side: $v_l.a<v_r.a<x.a$ is false, and $x.a<v_l.a<v_r.a$ is false. Therefore candidates in $C_r$ cannot affect pruning decisions for candidates in $C_l$, and candidates in $C_l$ cannot affect pruning decisions for candidates in $C_r$.
\end{proof}

\subsection{Proof of Lemma~\ref{lemma:prune-order}}\label{app:proof-prune-order}
\begin{proof}
Consider the left side $C_l$; the right side is symmetric. If $c_i,c_j\in C_l$ and $|c_i.a-x.a|<|c_j.a-x.a|$, then $c_i.a>c_j.a$. For $c_j$ to witness and prune edge $(x,c_i)$, it would need to lie between $c_i$ and $x$ in attribute order, i.e., $c_i.a<c_j.a<x.a$, which contradicts $c_j.a<c_i.a$. Thus a larger-gap left-side candidate cannot prune a smaller-gap left-side candidate. On the right side, the inequality reverses analogously and pruning would require $x.a<c_j.a<c_i.a$, which is impossible.
\end{proof}

\subsection{Proof of Theorem~\ref{thm:prunning-complexity}}\label{app:proof-pruning-complexity}
\begin{proof}
Let $\deg(v)$ be the out-degree of node $v$ in the RRNG. Constructing the exact RRNG with Algorithm~\ref{alg:RAP} uses $C=V$ and $m=\infty$ for every node. For a single node $v$, the algorithm scans $n$ candidates and compares each candidate against the retained set, whose final size is $\deg(v)$. The per-node cost is therefore bounded by $O(n\deg(v))$. Summing over all nodes gives
\[
  \sum_{v\in V} O(n\deg(v)) = O\!\left(n\sum_{v\in V}\deg(v)\right)=O(nM_{rrng}),
\]
where $M_{rrng}$ is the number of RRNG edges.
\end{proof}

\subsection{Proof of Theorem~\ref{thm:rnsg}}\label{app:proof-rnsg-build-time}
\begin{proof}
For each node $v$, Algorithm~\ref{alg:RAP} scans at most $|C|$ candidates and compares each candidate against the retained RNSG neighbors of $v$. Its cost is therefore $O(|C|\cdot \deg_{RNSG}(v))$. Summing this bound over all nodes gives
\[
  O\!\left(|C|\sum_v \deg_{RNSG}(v)\right)=O(|C|M_{rnsg}),
\]
where $M_{rnsg}$ is the number of edges in the final RNSG.
\end{proof}

\subsection{Proof of Theorem~\ref{thm:search-complex}}

\begin{lemma}[expected length of monotonic path]\label{thm:irng-exp-length}
For any start node $x\in V$ and any target object $y\in V$, the greedy monotone walk on $G$ from $x$ to $y$ has expected length
\[
\mathbb{E}\big[\mathrm{steps\ on}\ G\big]
= O\!\left(\frac{n^{1/d}\,\log n}{\Delta r}\right),
\qquad n:=|V|.
\]
where $\Delta r$ denotes the minimum difference in side lengths among any non-isosceles triangles in the space.
Moreover, if $x$ and $y$ are label-adjacent in an interval $I$ (no label strictly between $a(x)$ and $a(y)$ inside \(I\)), then the directed edge from $x$ to $y$ is retained in \(G[I]\) and the walk terminates in one step.
\end{lemma}
\begin{proof}
RRNG satisfies the monotonic searchability property proved in Theorem~\ref{thm:rrng-monotone}. The standard MSNET analysis for monotonic graph search~\cite{fu2017fast} bounds the expected length of a monotone walk by the number of distance-reducing cone crossings around the target, which is \(O(n^{1/d}\log n/\Delta r)\) under the same separation assumption. RRNG differs from MRNG only by restricting valid witnesses with the attribute-between condition; the monotonic path guaranteed by Theorem~\ref{thm:rrng-monotone} still follows distance-reducing directed edges. Therefore the same path-length bound applies. If $x$ and $y$ are label-adjacent inside $I$, no attribute-between witness exists, so the directed edge from $x$ to $y$ is retained and the walk terminates in one step.
\end{proof}

\begin{lemma}\label{lem:rrng-exp-deg}
Let the average degree of \yin{MRNG} be $D_r$~\cite{fu2017fast}. Under the random attribute-order model used in the RRNG size analysis, the expected average degree of RRNG is bounded by
\[
2D_r\left[1+\log\left(\frac{N-1}{2D_r}\right)\right],
\]
where $D_r$ is bounded by the maximum degree of \yin{MRNG}, a constant in fixed-dimensional Euclidean space.
\end{lemma}
\begin{proof}

Fix $x\in V_I$ and cover the unit sphere around $x$ by $N_d$ cones $\{C_j\}_{j=1}^{N_d}$ of half-angle $30^\circ$.
Split candidates into the \emph{left} side $\{a(\cdot)<a(x)\}$ and the \emph{right} side $\{a(\cdot)>a(x)\}$.
For a fixed cone $C$ and a fixed side $s\in\{\text{left},\text{right}\}$, let $M$ be the number of other points of $V_I$ that fall in $(C,s)$.
By (i)–(ii) and symmetry, $M\sim \mathrm{Bin}(n_I-1,p)$ with $p=\frac{1}{2N_d}$.

Under RRNG pruning, within $(C,s)$ we scan candidates in \emph{label order}. 
Because $s$ fixes the label direction (left/right), any earlier-scanned point is label-between a later one; and in a cone of half-angle $30^\circ$, if $z$ is no farther from $x$ than $y$, then $xy$ is the longest side in the triangle $(x,y,z)$, hence $z$ witnesses the removal of $xy$.
Therefore, a directed edge $x\to y$ is \emph{kept} iff $y$ is a \emph{record minimum} of the distance sequence w.r.t.\ the scanning order.
Let $R_M$ be the number of record minima among $M$ i.i.d.\ items under a random permutation independent of the distances (by (ii)); then
\[
\mathbb{E}[R_M\,|\,M=m]\;=\;H_m\;\triangleq\;\sum_{k=1}^{m}\frac{1}{k}\quad(m\ge 0,\ H_0:=0),
\]
a classical fact proved by $\Pr\{\text{the $k$-th is a new minimum}\}=1/k$ and linearity of expectation.

By linearity across cones and sides,
\[
\mathbb{E}\big[\deg^+_{G}(x)\big]
\;=\;\sum_{s\in\{\text{L},\text{R}\}}\sum_{j=1}^{N_d}\mathbb{E}[R_M]
\;=\;2N_d\,\mathbb{E}[H_M].
\]
For an upper bound that is closed-form and distribution-free in $M$, use $H_m\le 1+\log(1+m)$ for all $m\ge 0$ and the concavity of $\log$:
\[
\mathbb{E}[H_M]\;\le\;\mathbb{E}[1+\log(1+M)]
\;\le\;1+\log\big(1+\mathbb{E}[M]\big)
\;=\;1+\log\Big(1+\tfrac{n_I-1}{2N_d}\Big).
\]
Multiplying by $2N_d$ yields the claimed bound. 

\end{proof}

\begin{proof}
By Lemma~\ref{thm:irng-exp-length}, a monotone walk has expected length \(O(n^{1/d}\log n/\Delta r)\). Lemma~\ref{lem:rrng-exp-deg} bounds the expected number of retained RRNG outgoing neighbors per node by \(O(D_r\log n)\). In fixed-dimensional Euclidean space, \(D_r\) is treated as a constant, and the standard MRNG/MSNET analysis treats the separation factor \(\Delta r\) as a model constant for the expected-case bound~\cite{fu2017fast}. Therefore node-to-node monotone search on RRNG expands \(O(n^{1/d}\log^2 n)\) candidates in expectation. Since \(n^{1/d}\) is close to a constant for high-dimensional vector data, this is often written as \(O(\log^2 n)\).
\end{proof}


\subsection{Proof of Theorem~\ref{thm:space-complex}}
\begin{proof}
The space complexity of RRNG consists of two parts: the number of nodes and the number of edges. The number of nodes is $O(n)$. According to Lemma~\ref{lem:rrng-exp-deg}, we can derive that the number of edges in RRNG is $O(2nD_r[1+\log(\frac{N-1}{2D_r})])$. Furthermore, since $D_r$ is the average number of edges for MRNG, the expected number of edges for MRNG is $S_r = nD_r$. Thus, the number of edges can be expressed as $$O(2S_r[1+\log(\frac{N-1}{2D_r})]).$$ Simplifying this yields $O(S_r\log n)$. Since the number of edges in this expression far exceeds the number of vertices, we can treat the edge count and space complexity as equivalent. Therefore, the space complexity is also $O(S_r\log n)$.
\end{proof}

\comment{
\subsection{Proof of Lemma~\ref{lemma:prune-divide}}
\begin{proof}
We observe that for any $v_l \in C_l$, we have $v_l.a < x.a$, and for any $v_r \in C_r$, we have $v_r.a > x.a$. Taking any pair of nodes $v_l$ and $v_r$ from $C_l$ and $C_r$ respectively, we observe that $v_l.a < x.a < v_r.a$ must hold. Thus, neither of our pruning conditions $v_l.a < v_r.a < x.a$ nor $x.a < v_l.a < v_r.a$ holds. Therefore, for any $v_l \in C_l$, $v_r$ cannot affect its pruning. Similarly, for any $v_r \in C_r$, $v_l$ cannot affect its pruning.
\end{proof}
}

\comment{
\subsection{Proof of Lemma~\ref{lemma:prune-order}}
\begin{proof}
Since $c_i,c_j\in C$ have $v.a<x.a$, the inequality $|c_i.a-x.a|<|c_j.a-x.a|$ implies $c_i.a>c_j.a$ (points closer to $x$ on the left have larger attribute values).

Algorithm~\ref{alg:RAP} prunes within each side and only allows a witness to prune $(x,c_i)$ if it lies \emph{between} $c_i$ and $x$ in attribute order on that side, i.e., $c_i.a<c_j.a<x.a$ on the left. But we have $c_j.a<c_i.a$, so this betweenness condition fails. Therefore $(x,c_i)$ cannot be pruned by $c_j$.

For the right side ($v.a>x.a$), if $|c_i.a-x.a|<|c_j.a-x.a|$ then $c_i.a<c_j.a$, while pruning would require $x.a<c_j.a<c_i.a$, which is impossible. Hence the statement holds symmetrically.
\end{proof}
}

\subsection{Proof of Theorem~\ref{thm:RAP-gets-RRNG}}
\begin{proof}
Fix a center node \(x\). By Lemma~\ref{lemma:prune-divide}, candidates on the left and right of \(x.a\) can be considered independently, so it is sufficient to prove the claim for one side. Consider the right side; the left side is symmetric. The candidates are processed in increasing attribute gap from \(x\), so any valid attribute-between witness for a candidate \(y\) must have been processed earlier.

We prove by induction over this order that the retained set \(R\) maintained by Algorithm~\ref{alg:RAP} is exactly the RRNG neighbor set of \(x\) among the candidates processed so far. The first candidate on the side has no attribute-between witness, so it is retained by both Algorithm~\ref{alg:RAP} and the RRNG rule. For the induction step, consider the current candidate \(y\). If Algorithm~\ref{alg:RAP} prunes \(y\), then there is an already retained candidate \(z\in R\) such that
\[
\delta(x,z)<\delta(x,y),\qquad \delta(y,z)<\delta(x,y),
\]
and \(x.a<z.a<y.a\). By the induction hypothesis, \((x,z)\) is an RRNG edge, so \(z\) is a valid RRNG witness and \((x,y)\) is absent in RRNG. Conversely, suppose Algorithm~\ref{alg:RAP} keeps \(y\). Any RRNG witness \(z\) for pruning \((x,y)\) must lie between \(x\) and \(y\) in attribute order and therefore must have been processed earlier. By the induction hypothesis, such a witness would already be in \(R\), and Algorithm~\ref{alg:RAP} would have pruned \(y\), a contradiction. Hence no RRNG witness exists and \((x,y)\) is retained in RRNG. Applying the same argument to both sides proves that Algorithm~\ref{alg:RAP} with \(C=D\) and \(m=\infty\) yields exactly the RRNG.
\end{proof}

\comment{
\subsection{Proof of Theorem~\ref{thm:prunning-complexity}}
\begin{proof}
Following the algorithm's procedure, we construct the RRNG graph using algorithm \ref{alg:rnsg}, which requires pruning $n$ points. For each point, we examine a total of $n$ candidates. For each candidate, we scan the current set $R$ once. 
Since the algorithm continuously adds points to $R$ until it becomes the final neighborhood set, the pruning process satisfies $|R|\leq \textit{deg}(v)$.
Therefore, the time complexity of performing pruning on a single vertex is $O(\textit{deg}(v)n)$, and our overall time complexity can be derived as $O(n\sum{\textit{deg}(v)}) = O(\Delta n^2)$.
\end{proof}
}

\subsection{Entry Node Generation Details}\label{app:entry-demo}
\begin{algorithm}[t]
\setcounter{AlgoLine}{0}
\LinesNumbered
\caption{Entry Nodes Generation}
\label{alg:entry}
\small

\KwIn{Node set $V$}
\KwOut{A list of entry node sets $T$}

\SetKwFunction{GenEntry}{GenEntry}
\SetKwProg{Fn}{Function}{}{}

\Fn{\GenEntry{$V$}}{
    Calculate the centroid $c$ of the nodes in $V$\;
    Initialize $T \leftarrow \varnothing$\;
    Initialize $q \leftarrow \varnothing$\;
    Sort $V$ by nondecreasing attribute values to obtain $\langle v_1, v_2, \dots, v_n \rangle$\;

    \For{$i \leftarrow 1$ \KwTo $n$}{
        Remove all node $x$ in $q$ such that $\delta(x,c)>\delta(v_i,c)$\;
        $q \rightarrow q \cup \{v_i\}$\;
        $T \rightarrow T \cup \{q\}$\;
    }

    \Return $T$\;
}
\end{algorithm}

Algorithm~\ref{alg:entry} implements the header entry policy used in Section~\ref{subsec:query-pro-rnsg}. When processing $v_i$, any earlier node in $q$ that is farther from the centroid $c$ than $v_i$ can be removed, because $v_i$ is at least as good under the centroid-entry proxy for every later right endpoint. The current $q$ is then stored for $v_i$.

\begin{figure}[t]
  \centering
  \includegraphics[width=0.50\linewidth]{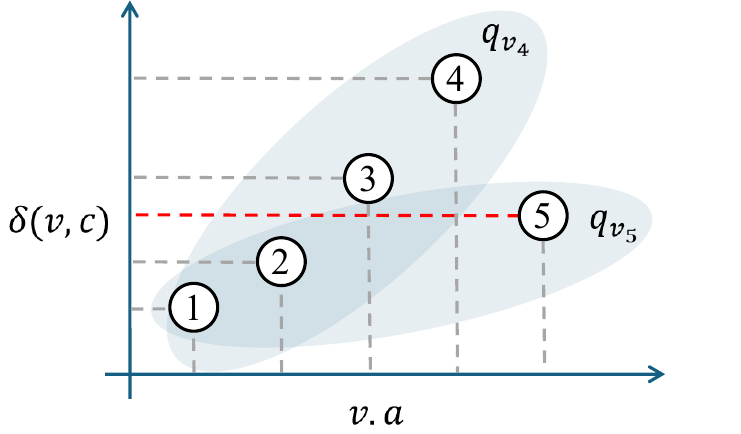}
  \caption{Illustration of entry node set generation.}
  \label{fig:entry-demo}
\end{figure}

\begin{example}
In Fig.~\ref{fig:entry-demo}, the horizontal axis denotes the attribute $v.a$ and the vertical axis is the distance $\delta(v,c)$ from each node to the global centroid $c$.
At $v_4$, the entry set is $q_{v_4}=\{v_1,v_2,v_3,v_4\}$.
When Algorithm~\ref{alg:entry} processes $v_5$, its distance satisfies $\delta(v_5,c)<\delta(v_3,c)$ and $\delta(v_5,c)<\delta(v_4,c)$, so the algorithm removes $v_3$ and $v_4$ and obtains $q_{v_5}=\{v_1,v_2,v_5\}$.
For any query whose right endpoint is $v_5.a$, $v_5$ is a better entry candidate than $v_3$ and $v_4$ under the centroid-distance proxy, hence both $v_3$ and $v_4$ are no longer needed.
\end{example}

\comment{
\subsection{Proof of Lemma~\ref{lem:selection rate}}
\begin{proof}
    
\end{proof}
}

\subsection{Proof of Theorem~\ref{thm:rnsg-induce}}
\begin{proof}
Let \(G=(V,E)\) be the RNSG built on \(D\). For every node \(x\), let \(C_x\) be the candidate set generated by Algorithm~\ref{alg:rnsg} before Algorithm~\ref{alg:RAP} is applied. For interval \(I=[a_l,a_r]\), let \(V_I=\{x\in V:x.a\in I\}\), and let \(H[I]\) be the graph obtained by applying Algorithm~\ref{alg:RAP} to every \(x\in V_I\) using the restricted candidate set \(C_x\cap V_I\) and the same degree bound \(m\).

It remains to prove that pruning commutes with this restriction. Algorithm~\ref{alg:RAP} splits \(C_x\) into the left and right sides of \(x.a\), which are independent by Lemma~\ref{lemma:prune-divide}. Consider the right side; the left side is symmetric. The right-side candidates are processed in nondecreasing attribute gap from \(x\). For any in-range right-side candidate \(y\), every candidate processed before \(y\) has attribute in \((x.a,y.a)\), and is therefore also in \(V_I\). Hence the prefix of candidates that can affect the decision on \(y\) is identical in the full run and in the run restricted to \(V_I\).

The geometric witness test is also identical. If a kept candidate \(w\) prunes \((x,y)\), then the RRNG witness condition requires \(x.a<w.a<y.a\). Because \(x,y\in V_I\) and \(I\) is an interval, \(w\in V_I\). Conversely, every in-range witness available in the restricted run is also available in the full run. Therefore each in-range right-side candidate receives the same prune-or-keep decision in both runs. The degree bound \(m/2\) does not break the argument: before any in-range candidate is processed, all previously processed same-side candidates are in range, so the budget is filled at exactly the same point in the full and restricted runs. After the first out-of-range right-side candidate is processed, no later right-side candidate can be in range. Thus the retained right-side set, restricted to \(V_I\), equals the retained right-side set produced inside \(V_I\). The same holds for the left side.

Finally, Algorithm~\ref{alg:RAP} merges the two sides. Therefore every \(x\in V_I\) has the same neighbor set in the induced subgraph \(G[I]\) as in \(H[I]\). Hence \(G[I]=H[I]\).

\end{proof}

\subsection{Proof of Lemma~\ref{lem:entrysize}}
\begin{proof}

Let $V=\langle v_1,\dots,v_n\rangle$ be sorted in non-decreasing attribute order as in Algorithm~\ref{alg:entry}, and fix a right endpoint $a_r$. Let $v_r$ be the node with the largest attribute value not exceeding $a_r$. For each $t\le r$, define the query range $I_t=[v_t.a,a_r]$ and let $c$ denote the centroid of $V$, with $dis_i \triangleq \delta(v_i,c)$.

Under the independence assumption between embeddings and attributes, the sequence $(dis_1,\dots,dis_r)$ consists of independent and identically distributed (i.i.d.) random variables (ties occur with probability zero or are broken arbitrarily). The entry node for range $I_t$ is defined as: \[
e_t \;\triangleq \; \arg\min_{1\le s\le t} dis_s.
\]

As $t$ increases from $1$ to $r$, $e_t$ changes precisely when $dis_t < \min_{s<t} dis_s$, i.e., when $dis_t$ is a new record minimum.

For i.i.d. continuous random variables, the probability that $dis_t$ is a new record minimum is $1/t$. Therefore, by linearity of expectation:
\[
\mathbb{E}\!\left[\#\{e_t:\,1\le t\le r\}\right]
=\sum_{t=1}^{r}\frac{1}{t}
=H_r
\le 1+\ln r
=O(\log n).
\]
This completes the proof.
\end{proof}

\begin{figure}
  \centering
  \includegraphics[width=0.95\linewidth]{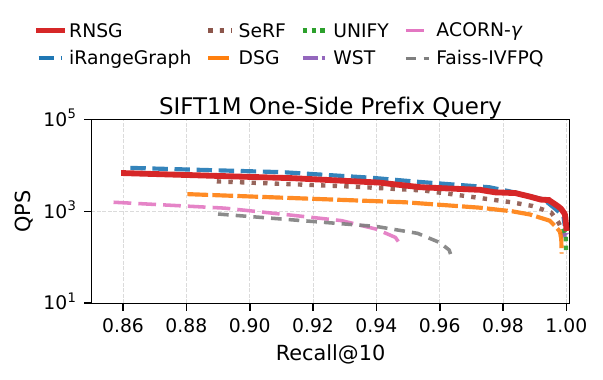}
  \caption{One-sided prefix-range query experiment on SIFT1M.}
  \label{fig:sift-one-side}
\end{figure}

\begin{figure}
  \centering
  \includegraphics[width=\linewidth]{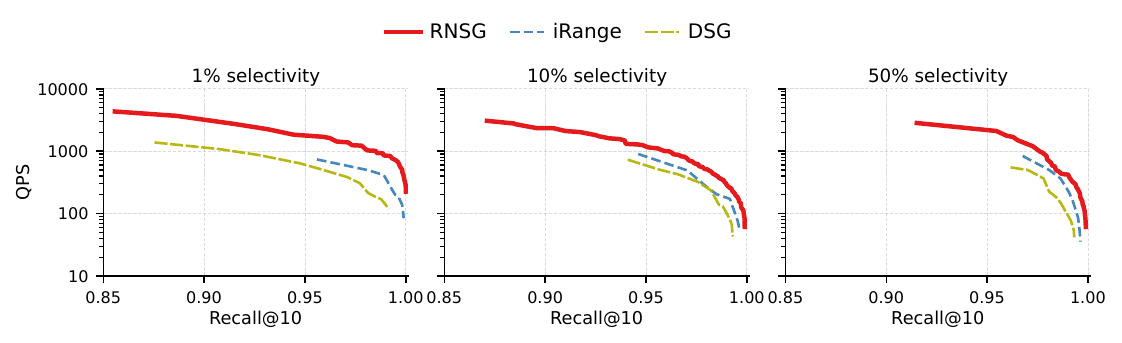}
  \caption{Large-scale recall-QPS curves on SpaceV10M. RNSG is compared with iRangeGraph and DSG under 1\%, 10\%, and 50\% selectivity workloads.}
  \label{fig:spacev10m-large-scale}
\end{figure}

\subsection{Exp-9: One-sided prefix-range query experimen.}
Some baselines may be especially favorable under one-sided interval predicates, so Fig.~\ref{fig:sift-one-side} separately evaluates prefix-range queries on SIFT1M. Under this special query form, SeRF becomes much stronger than in the standard interval workloads, and iRangeGraph retains a clear advantage at low and medium recall. RNSG does not use a one-sided query-specific optimization, but still gives the second-best mid-recall trade-off after iRangeGraph and remains competitive at high recall, reaching 0.984 recall at 2507 QPS and 0.9995 recall at 882 QPS, which is $0.96\times$ and $1.79\times$ of iRangeGraph QPS, respectively. This confirms that RNSG can handle one-sided range predicates effectively while retaining its general interval-query design.

\subsection{Exp-10: Comparing with baseline on SpaceV10M}
Fig.~\ref{fig:spacev10m-large-scale} evaluates whether the single-graph design remains effective beyond the 1M-object benchmarks. We compare RNSG with iRangeGraph and DSG on SpaceV10M using 1K recall@10 queries and the same 1\%, 10\%, and 50\% selectivity levels used in the main high-selectivity study. Each method uses an explicit construction profile and then sweeps query-side search effort to form recall-QPS curves: RNSG uses KNNG degree $128$, $ef_\mathrm{attribute}=16384$, and $m=400$; iRangeGraph uses $M=40$ and $ef_\mathrm{construction}=100$; DSG uses $M=16$, $ef_\mathrm{max}=1000$, and $ef_\mathrm{construction}=160$.

The large-scale curves preserve the main trend. At recall at least 0.98, RNSG reaches 1028.74, 505.30, and 653.81 QPS for 1\%, 10\%, and 50\% selectivity, respectively. These points are $2.08\times$, $2.48\times$, and $1.32\times$ faster than iRangeGraph, and $4.82\times$, $2.16\times$, and $2.91\times$ faster than DSG. At recall at least 0.99, RNSG remains $3.50\times$, $1.61\times$, and $1.88\times$ faster than iRangeGraph, and $7.10\times$, $3.18\times$, and $4.47\times$ faster than DSG. The deployed RNSG index is 4.20 GiB, compared with 21.20 GiB for iRangeGraph and 59.22 GiB for DSG. Its full construction time is 2283.31 seconds, reported as 1945.43 seconds for KNNG plus 337.88 seconds for the RNSG-specific build stage, compared with 16955.00 seconds for iRangeGraph and 11346.82 seconds for DSG.
\fi
\end{document}